\documentclass[journal]{IEEEtran}
\usepackage{amsmath,amsfonts,amssymb,bm}
\usepackage{algorithmic}
\usepackage{algorithm}
\usepackage{array}
\usepackage[subrefformat=parens, labelformat=parens, caption=false,font=footnotesize]{subfig}
\usepackage{textcomp}
\usepackage{stfloats}
\usepackage{url}
\usepackage{verbatim}
\usepackage{graphicx}
\usepackage{cite}
\usepackage{color}
\usepackage{xcolor}
\usepackage{framed}
\usepackage{hyperref}

\definecolor{darkgreen}{rgb}{0.0, 0.5, 0.0}

\definecolor{pccolor}{rgb}{0.0, 0.0, 0.5}

\definecolor{pcQcolor}{rgb}{0.0, 0.0, 1.0}

\usepackage{amsthm}
\theoremstyle{definition}
\theoremstyle{definition}
\theoremstyle{definition}\newtheorem{theorem}{Theorem}
\theoremstyle{definition}
\newtheorem{define}{Definition}
\newtheorem{example}{Example}

\newtheorem{remark}{Remark}
\def\main{1}\ifx\main\undefined
\ifx\singleSec\undefined
\documentclass[journal,onecolumn,12pt]{IEEEtran}
\usepackage{graphicx,cite}
\usepackage{amssymb,amsmath,yhmath}

\begin{document}
\fi
\fi

\newcommand{\ben}{\begin{eqnarray}}
\newcommand{\een}{\end{eqnarray}}

\DeclareMathOperator{\diag}{diag}
\DeclareMathOperator{\E}{E}
\newcommand{\vect}{\mathrm{vec}}

\newcommand{\bl}{\color{blue}}
\newcommand{\rd}{\color{red}}
\newcommand{\bk}{\color{black}}

\newcommand{\real}{\mathrm{Re}}
\newcommand{\imag}{\mathrm{Im}}

\newcommand{\va}{{\bf a}}
\newcommand{\vb}{{\bf b}}
\newcommand{\vc}{{\bf c}}
\newcommand{\vd}{{\bf d}}
\newcommand{\ve}{{\bf e}}
\newcommand{\vf}{{\bf f}}
\newcommand{\vg}{{\bf g}}
\newcommand{\vh}{{\bf h}}
\newcommand{\vi}{{\bf i}}
\newcommand{\vj}{{\bf j}}
\newcommand{\vk}{{\bf k}}
\newcommand{\vl}{{\bf l}}
\newcommand{\vm}{{\bf m}}
\newcommand{\vn}{{\bf n}}
\newcommand{\vo}{{\bf o}}
\newcommand{\vp}{{\bf p}}
\newcommand{\vq}{{\bf q}}
\newcommand{\vr}{{\bf r}}
\newcommand{\vs}{{\bf s}}
\newcommand{\vt}{{\bf t}}
\newcommand{\vu}{{\bf u}}
\newcommand{\vv}{{\bf v}}
\newcommand{\vw}{{\bf w}}
\newcommand{\vx}{{\bf x}}
\newcommand{\vy}{{\bf y}}
\newcommand{\vz}{{\bf z}}
\newcommand{\zv}{{\bf 0}}

% caligraphy fonts
\newcommand{\cA}{{\cal A}}
\newcommand{\cB}{{\cal B}}
\newcommand{\cC}{{\cal C}}
\newcommand{\cD}{{\cal D}}
\newcommand{\cE}{{\cal E}}
\newcommand{\cF}{{\cal F}}
\newcommand{\cG}{{\cal G}}
\newcommand{\cH}{{\cal H}}
\newcommand{\cI}{{\cal I}}
\newcommand{\cJ}{{\cal J}}
\newcommand{\cK}{{\cal K}}
\newcommand{\cL}{{\cal L}}
\newcommand{\cM}{{\cal M}}
\newcommand{\cN}{{\cal N}}
\newcommand{\cO}{{\cal O}}
\newcommand{\cP}{{\cal P}}
\newcommand{\cQ}{{\cal Q}}
\newcommand{\cR}{{\cal R}}
\newcommand{\cS}{{\cal S}}
\newcommand{\cT}{{\cal T}}
\newcommand{\cU}{{\cal U}}
\newcommand{\cV}{{\cal V}}
\newcommand{\cW}{{\cal W}}
\newcommand{\cX}{{\cal X}}
\newcommand{\cY}{{\cal Y}}
\newcommand{\cZ}{{\cal Z}}

% Greek letters
\newcommand{\bs}[1]{\boldsymbol{#1}}
\newcommand{\gS}{\boldsymbol\Sigma}
\newcommand{\ga}{\alpha}
\newcommand{\gb}{\beta}
\newcommand{\ggm}{\gamma}
\newcommand{\gd}{\delta}
\newcommand{\gl}{\lambda}
\newcommand{\gr}{\rho}
\newcommand{\gs}{\sigma}
\newcommand{\gq}{\theta}
\newcommand{\gQ}{\bs{\theta}}

% Dots
\newcommand{\dd}{\ddots}
\newcommand{\dc}{\cdots}
\newcommand{\dv}{\vdots}
\newcommand{\da}{\adots}

% Sets of Numbers
\newcommand{\bC}{{\mathbb{C}}}
\newcommand{\bR}{{\mathbb{R}}}
\newcommand{\bZ}{{\mathbb{Z}}}
\newcommand{\bQ}{{\mathbb{Q}}}
\newcommand{\bN}{{\mathbb{N}}}

% Logical
\newcommand{\ra}{\rightarrow}
\newcommand{\Ra}{\Rightarrow}

\newenvironment{mat}[1]{\left[\begin{array}{#1}}{\end{array}\right]}
\newcommand{\col}[1]{\begin{mat}{c}#1\end{mat}}
\newcommand{\mx}[1]{\begin{mat}{ccccccc}#1\end{mat}}
\newcommand{\mf}[2]{\begin{mat}{#1}#2\end{mat}}
\newcommand{\set}[1]{\left\{#1\right\}}
\newcommand{\setc}[2]{\left\{{#1}{ : }{#2}\right\}}
\newcommand{\ip}[2]{\left<#1\right>_{\{#2\}}}
\newcommand{\toe}[2]{{\cal T}_{#1}^{#2}}

\newcommand{\ceil}[1]{\left\lceil#1\right\rceil}
\newcommand{\floor}[1]{\left\lfloor#1\right\rfloor}
\newcommand{\pp}[1]{\left(#1\right)}
\newcommand{\pb}[1]{\left[#1\right]}

%% Tensor related commands

\newcommand{\mz}[3]{{{#1}}^{\set{#2}}_{\set{#3}}}  % matricize
\newcommand{\ttv}[3]{{#1}\bar\times_{#3}{#2}}
\newcommand{\ttm}[3]{{#1}\times_{#3}{#2}}
\newcommand{\ttt}[4]{\left<{#1},{#2}\right>_{\left\{#3;#4\right\}}}
\newcommand{\tpm}[2]{{#1}^{\set{#2}}}
\renewcommand{\tt}[2]{{\cT_{#1}^{#2}}}
\newcommand{\ttq}[3]{{\cT_{#1,#2}^{#3}}}
\newcommand{\ttd}[4]{{\cT_{#1,#2}^{#3,#4}}}

%% Usually used
\newcommand{\MQ}{{M_Q}}
\newcommand{\PQ}{{P_Q}}
\newcommand{\DQ}{{D_Q}}
\newcommand{\deq}{\triangleq}
\newcommand{\aeq}{\approx}
\renewcommand{\t}{\times}
\newcommand{\eq}[2]{\begin{equation}\label{#1}#2\end{equation}}
\newcommand{\eqa}[1]{\begin{eqnarray*}#1\end{eqnarray*}}

% Mbox functions
\newcommand{\rk}{\mbox{rank }}

% General function
\newcommand{\gf}[2]{\mbox{#1}\left(#2\right)}

\newcommand{\mA}{{\bf A}}
\newcommand{\mB}{{\bf B}}
\newcommand{\mC}{{\bf C}}
\newcommand{\mD}{{\bf D}}
\newcommand{\mE}{{\bf E}}
\newcommand{\mF}{{\bf F}}
\newcommand{\mG}{{\bf G}}
\newcommand{\mH}{{\bf H}}
\newcommand{\mI}{{\bf I}}
\newcommand{\mJ}{{\bf J}}
\newcommand{\mK}{{\bf K}}
\newcommand{\mL}{{\bf L}}
\newcommand{\mM}{{\bf M}}
\newcommand{\mN}{{\bf N}}
\newcommand{\mO}{{\bf O}}
\newcommand{\mP}{{\bf P}}
\newcommand{\mQ}{{\bf Q}}
\newcommand{\mR}{{\bf R}}
\newcommand{\mS}{{\bf S}}
\newcommand{\mT}{{\bf T}}
\newcommand{\mU}{{\bf U}}
\newcommand{\mV}{{\bf V}}
\newcommand{\mW}{{\bf W}}
\newcommand{\mX}{{\bf X}}
\newcommand{\mY}{{\bf Y}}
\newcommand{\mZ}{{\bf Z}}

\newcommand{\DmY}{{\Delta{\bf Y}}}
\newcommand{\DmR}{{\Delta{\bf R}}}
\newcommand{\DcU}{{\Delta{\cal U}}}
\newcommand{\DcY}{{\Delta{\cal Y}}}
\newcommand{\Dvh}{{\Delta{\bf h}}}
\newcommand{\Dvy}{{\Delta{\bf y}}}
\newcommand{\RxQ}{{\mR_{\mX,Q}}}
\newcommand{\YQJ}{{\cY_Q^{(J)}}}
\newcommand{\XQJ}{{\cX_Q^{(J)}}}
\newcommand{\DmUn}{\Delta\mU_n}
\newcommand{\Rss}{\mR_{\tilde\vs\tilde\vs^\dagger}}

\newtheorem{def1}{Definition}
\newtheorem{lem}{Lemma}
\newtheorem{thm}{Theorem}
\newtheorem{cor}{Corollary}

\newcommand{\tplz}[4]{\mathcal{T}_{#1}^{(#2,#3)}\left(#4\right)}
\newcommand{\tpz}[2]{\mathcal{T}_{#1}\left(#2\right)}

\newcommand{\ignore}[1]{}

\ifx\main\undefined
\ifx\singleSec\undefined
newcommands.tex (ver 1.0, 2015/12/18)\\

\section{Usage}
To use this file, type
\begin{verbatim}
\def\main{1}
\input{newcommands}
\end{verbatim}
in your main file. \\
Here, the first line declares that your file is the main file you're working on so that the demo text shown here will not show up in your main text.

\section{Demo}
This file provides the following notations:

\subsection{Symbols}
Vectors:
$\va\vb\vc\vd\ve\vf\vg\vh\vi\vj\vk\vl\vm\vn\vo\vp\vq\vr\vs\vt\vu\vv\vw\vx\vy\vz$.\\

Matrices:
$\mA\mB\mC\mD\mE\mF\mG\mH\mI\mJ\mK\mL\mM\mN\mO\mP\mQ\mR\mS\mT\mU\mV\mW\mX\mY\mZ$.
\\

Tensors
$\cA\cB\cC\cD\cE\cF\cG\cH\cI\cJ\cK\cL\cM\cN\cO\cP\cQ\cR\cS\cT\cU\cV\cW\cX\cY\cZ$\\

Greek
$ \ga \gb \ggm \gd\gl \gr \gs$\\

BoldGreek
$\gS \bs{\ga \gb \ggm \gd\gl \gr \gs}$

Dots
$\dd\dc\dv\da$

\subsection{Environments}
Set: $\set{1,2,3}$\\
Set with codition: $\setc{\vx}{\mA\vx=\zv}$\\
Pair of parenthesis: $\pp{1+1}$, $$\pp{{1\over 2}}$$.\\
Pair of brackets: $\pb{\vv}_1$.

\subsection{Tensors}
TTV: $\ttv{\cA}{\vv}{1}.$\\
TTM: $\ttm{\cA}{\mM}{2}.$\\
TTT: $\ttt{\cA}{\cB}{1,2}{3,4}.$\\
Tensor permutation: $\tpm{\cA}{1,3,2}, \tpm{ \ttt{\cA}{\cB}{1,2}{3,4}}{1,3,2} $.

\end{document}

\fi
\fi

\begin{document}

\title{Fundamental Limits of MIMO ISAC: An Antenna Array Architecture Perspective}

\author{Po-Chih Chen,~\IEEEmembership{Member,~IEEE,} Ming-Chun Lee,~\IEEEmembership{Member,~IEEE,} Yu-Chih Huang,~\IEEEmembership{Senior Member,~IEEE}
        % <-this % stops a space
\thanks{This work has been submitted to the IEEE for possible publication. Copyright may be transferred without notice, after which this version may no longer be accessible.}% <-this % stops a space
\thanks{A preliminary version of this work, which considers only the angle-domain model without waveform-specific analysis, will be submitted to a conference.}% <-this % stops a space
\thanks{This work was supported by the National Science and Technology Council of Taiwan under grant NSTC 114-2224-E-A49-002 and grant NSTC 115-2221-E-A49-077.}% <-this % stops a space
\thanks{The authors are with the Institute of Communications Engineering, National Yang Ming Chiao Tung University, Hsinchu City 300093, Taiwan (email: pochih@nycu.edu.tw; mingchunlee@nycu.edu.tw; jerryhuang@nycu.edu.tw).}% <-this % stops a space
% \thanks{Manuscript received April 19, 2021; revised August 16, 2021.}
}

% The paper headers
% \markboth{Journal of \LaTeX\ Class Files,~Vol.~14, No.~8, August~2021}%
% {Shell \MakeLowercase{\textit{et al.}}: A Sample Article Using IEEEtran.cls for IEEE Journals}

% \IEEEpubid{0000--0000/00\$00.00~\copyright~2021 IEEE}
% Remember, if you use this you must call \IEEEpubidadjcol in the second
% column for its text to clear the IEEEpubid mark.

\maketitle

\begin{abstract}
This paper investigates the fundamental limits of MIMO integrated sensing and communications (ISAC) systems, specifically comparing sparse arrays (SAs) against conventional uniform linear arrays (ULAs).
A unified theoretical analysis of ergodic channel capacity and the Cram\'{e}r--Rao bound (CRB) for angle estimation is developed while accounting for the array geometry. Utilizing the framework of stochastic majorization, the study reveals that SAs consistently outperform ULAs by creating a more ``uniform" spatial eigenvalue distribution, which decorrelates the multipath environment and increases communication capacity. For sensing, the paper proves that the angle CRB is inversely proportional to the array's second-order central moment of antenna positions asymptotically, demonstrating that SAs achieve superior accuracy---improving by up to the square of the number of antennas---due to their increased physical aperture.
These analyses and conclusions are demonstrated to be also valid for MIMO ISAC systems employed with the modern waveforms OTFS and OFDM, suggesting that spatial geometry, rather than waveform, is the primary driver of fundamental performance gains in the spatial dimension.
\end{abstract}

\begin{IEEEkeywords}
ISAC, sparse arrays, capacity, Cram\'{e}r--Rao bound, OTFS, OFDM.
\end{IEEEkeywords}

\section{Introduction}
\IEEEPARstart{T}{he} convergence of sensing and communications into a unified framework, known as integrated sensing and communications (ISAC) \cite{ZhangRahman,Cui,LiuCui,Lu}, has emerged as a cornerstone for next-generation wireless networks (6G). By sharing the same time-frequency resources and hardware infrastructure, ISAC systems promise significantly higher spectral and hardware efficiency.
Meanwhile, multiple-input-multiple-output (MIMO) technology, which provides spatial multiplexing gain and diversity gain over single-antenna systems, is a key technology for modern wireless communications. Consequently, MIMO ISAC has also attracted considerable research interest. Recently, MIMO ISAC systems have been extensively studied in various aspects, such as novel target parameter estimation methods \cite{XiaoSwindlehurst} and beamforming design \cite{Liao2025Design}.

One aspect that is crucial but historically less considered for MIMO ISAC is the antenna array geometry.
While the uniform linear array (ULA) is the de facto standard in communication systems due to its simplicity, the sensing community has long recognized the advantages of sparse arrays (SA) \cite{LiuCRB,ChenRank,ChenHybridCBS1}, such as nested arrays \cite{PiyaPPVnested}, coprime arrays \cite{PPVCoprimePaper}, and minimum redundancy arrays (MRA) \cite{Moffet}. By exploiting nonuniform element spacing, SAs can provide higher angular resolution and more degrees of freedom than a ULA with the same number of physical elements \cite{LiuCRB}. Despite these benefits in sensing, the impact of sparse geometries on communication capacity remains under-explored, as pointed out in the literature \cite{LiMin}, particularly in the context of joint ISAC waveforms.

While SAs have been studied in the context of precoding and antenna selection \cite{Wu2025Correlation,LiuSwindlehurst,Sankar}, the resulting highly non-convex optimization problems preclude a clear analytical characterization of the performance gains of SAs over ULAs.
The potential benefits of SAs for communications in terms of higher capacity and lower inter-user interference are shown in \cite{WangZeng} and analyzed via the effective degree of freedom (DoF) in \cite{WangFeng}, but the results therein are limited to line-of-sight dominant channels and sparse ULAs (ULAs with antenna spacing larger than half wavelength).
A hybrid processing framework for sparse MIMO ISAC with physical array based communication and difference coarray based sensing is proposed in \cite{Min}.
With the main lobe width and side lobe height characterized, it is shown that despite undesired grating lobes, SAs can enhance communication due to their narrower main lobe widths than ULAs \cite{Min}.
In these works \cite{LiuSwindlehurst,Sankar,WangZeng,WangFeng,Min}, specific modern waveforms like orthogonal time frequency space (OTFS) and orthogonal frequency division multiplexing (OFDM) are not explicitly considered in their ISAC models.

Waveform design is another problem in ISAC systems. Due to the abundant use of OFDM in current 4G and 5G networks, OFDM-based ISAC systems have been particularly widely studied \cite{XiaoSwindlehurst,Liao2025Design,Dai}.
However, to deploy ISAC in high-dynamics scenarios, such as vehicle-to-everything (V2X) communications \cite{Du} and low-earth orbit (LEO) satellite networks \cite{Kodheli}, OTFS \cite{Raviteja2018,Raviteja2019,WeiYuan} has been proposed as a robust alternative to OFDM to address the challenges posed by severe Doppler shifts and rapidly time-varying channels. For instance, transmit beamforming and symbol detection methods are developed in an OTFS ISAC transmission scheme for vehicular networks in \cite{Yuan}.
Evaluating the impact of array architecture on MIMO ISAC systems requires confirming that our conclusions remain valid when deployed with these modern waveforms.

In this paper, we investigate the fundamental limits of MIMO ISAC systems employing sparse antenna arrays.
To reveal the essential effects of array geometry on the ISAC system, we first focus on a snapshot-based spatial MIMO model.
We provide a unified analysis of the ergodic channel capacity and the Cram\'{e}r--Rao bound (CRB) for angle estimation in MIMO ISAC systems based on the geometric multipath MIMO channel model.
Indeed, modern ISAC systems are increasingly envisioned to operate in the millimeter-wave (mmWave) bands to leverage the vast available bandwidth for high-resolution sensing and high-rate communications \cite{Du,Cui,Lu}. Unlike lower frequency bands, mmWave channels exhibit a sparse scattering nature where signals propagate through a limited number of dominant paths. Consequently, the geometric multipath MIMO channel model---which explicitly accounts for the physical angles of these discrete paths---is widely adopted \cite{Heath,ChenHybridCBS2,XiaoSwindlehurst} to accurately characterize such systems.
By employing a geometric MIMO channel model, we can map the physical layout of the antenna array, whether uniform or sparse, directly onto the spatial correlation and eigenvalue distribution of the channel, thereby revealing the fundamental limits of each array geometry.

To provide a rigorous comparison between array geometries that transcends specific channel realizations, we employ the framework of stochastic majorization \cite{MarshallMajorization}. Our analysis is based on a majorization criterion applied to the eigenvalues of a spatial correlation matrix. This matrix is efficiently computed directly from the antenna locations and the probability distributions of the directions of departure (DODs) and directions of arrival (DOAs), without the need of exhaustive Monte Carlo simulations. This analytical tractability enables us to show that the majorization criterion is satisfied when comparing SAs to ULAs under most practical channel conditions.
Physically, the squared Frobenius norm of this spatial correlation matrix is intrinsically linked to the beam pattern of the array and the effective DoF \cite{WangFeng,Muharemovic} of the channel matrix. By analyzing the eigenvalues of the expected Gram matrix of the channel, we demonstrate that a ULA's compact structure leads to highly correlated channels and ``peaked'' eigenvalues. By contrast, SAs yield a more ``uniform'' eigenvalue distribution, which corresponds to a broader distribution of spatial energy. This spatial whitening effect indicates that sparse geometries can better decorrelate the multipath environment, thereby maximizing the fundamental communication capacity.

Regarding sensing performance, we focus on the fundamental limits of angle estimation of targets.
By casting the received sensing signal into the standard DOA estimation framework, we derive the \textit{conditional CRB} \cite{LiuCRB,StoicaNehoraiMusicMLCRB,StoicaStochasticCRB} for the angle estimates.
Specifically, we derive the Fisher information matrix for the unknown parameters and take the Schur complement of the block corresponding to the angle parameters.
By isolating the spatial contribution to sensing accuracy, we prove that the angle CRB is inversely proportional to the array's second-order central moment of antenna positions. This provides a clear and elegant insight: the superior sensing performance of SAs is a direct consequence of their increased physical aperture, even when the number of antennas remains constant across different waveforms.

Finally, we extend our analysis to MIMO ISAC models with two modern waveforms OTFS and OFDM. When these waveforms are considered, ranges and Doppler shifts naturally come into the model, but we still focus on angle parameters for sensing. We demonstrate that the main results of capacity and angle CRB comparison between SAs and ULAs are also valid for MIMO-OTFS and MIMO-OFDM ISAC systems.

Our main contributions are summarized as follows:

\begin{itemize}
    \item We establish a unified theoretical framework to analyze the ergodic channel capacity of MIMO ISAC systems. By leveraging the theory of stochastic majorization, we show that SAs consistently outperform conventional ULAs. We show that SAs induce a ``spatial whitening effect" that decorrelates the multipath channel, leading to a more uniform distribution of spatial eigenvalues and thus significantly higher capacity. The capacity gain is larger if the multiple paths are more clustered in the angle domain.
    \item For the sensing aspect, we derive the closed-form CRB for angle estimation. Our analysis reveals that the sensing accuracy is fundamentally governed by the second-order central moment of antenna positions. We demonstrate that typical SAs like nested arrays \cite{PiyaPPVnested} and coprime arrays \cite{PPVCoprimePaper} can achieve an accuracy gain proportional to the square of the number of sensing receive antennas compared to ULAs, due to their significantly larger physical aperture. 
    % This provides a theoretical foundation for designing high-precision sensing systems without increasing the number of RF chains.
    \item With a unified analysis, we show that the derived advantages in capacity and angle CRB of SAs over ULAs can be extended to MIMO-OTFS and MIMO-OFDM ISAC systems. While prior work like \cite{Tao} has studied the channel capacity of MIMO-OTFS systems, it lacks a discussion on the impact of array architectures.
\end{itemize}

% \textit{Paper outline:}

\section{System Models for MIMO ISAC Systems with Array Structure}
\label{sec:model}

In this section, we describe the necessary system and signal models of MIMO ISAC systems required for the subsequent analysis.
In general, we consider an ISAC transmitter that simultaneously supports communication with a receiver and monostatic sensing. The system consists of $N_{\mathrm{t}}$ transmit antennas and $N_{\mathrm{r}}$ communication receive antennas, and $N_{\mathrm{s}}$ monostatic radar receive antennas are equipped separately at the transmitter, denoted as sensing receiver, for monostatic sensing.
In addition, we consider $D_{\mathrm{c}}$ scattering paths for communication and $D_{\mathrm{s}}$ paths for sensing targets, and the communication and sensing channels are constructed by these paths with different DODs and DOAs, as illustrated in Fig. \ref{fig:MIMO_channel}.
We assume the scattering paths and their DODs and DOAs are time-invariant within a coherent processing interval (CPI) for the analysis. 
% However, it should be stressed that while the parameters are considered time-invariant, the channel coefficients could still be time-varying due to the occurrence of Doppler effect.
We denote $\theta_{\mathrm{t},i}$ and $\theta_{\mathrm{r},i}$ as the angle of scatterer $i$ with respect to the transmitter and communication receiver, respectively.
To show the fundamental effects of sparse arrays on the system, we mainly consider a snapshot-based spatial MIMO model. In Sec. \ref{sec:waveforms}, we will show that the main conclusions derived from the MIMO model still hold when two common modern waveforms OTFS and OFDM are incorporated into the model.

% \vspace{1mm}
\subsection{Communication Signal Model}

The received communication signal for $K$ independent snapshots within a CPI of a MIMO system is given by
\begin{equation}
    \mathbf{y}_{\mathrm{c}}[k] = \mathbf{H}_{\mathrm{c}}[k] \mathbf{x}[k] + \mathbf{w}_{\mathrm{c}}[k],
\end{equation}
for $k = 1, 2, \dots, K$, where
\begin{equation}
    \mathbf{H}_{\mathrm{c}}[k] = \sum_{i=1}^{D_{\mathrm{c}}} \alpha_{\mathrm{c},i}[k] \mathbf{a}_{\mathrm{r}}(\omega_{\mathrm{r},i})\mathbf{a}_{\mathrm{t}}^H(\omega_{\mathrm{t},i}) \label{eq:H_c_k}
\end{equation}
is the geometric MIMO communication channel matrix for snapshot $k$, $\alpha_{\mathrm{c},i}[k]$ is the channel coefficient of scatterer $i$; $\mathbf{x}[k] \in \mathbb{C}^{N_{\mathrm{t}} \times 1}$ is the transmitted symbols; $\mathbf{w}_{\mathrm{c}}[k]$ is the additive complex white Gaussian noise (AWGN) at the communication receiver; $\mathbf{a}_{\mathrm{r}}(\omega_{\mathrm{r},i})$ and $\mathbf{a}_{\mathrm{t}}(\omega_{\mathrm{t},i})$ are the receive and transmit steering vectors under DOA $\omega_{\mathrm{r},i}$ and DOD $\omega_{\mathrm{t},i}$ of scatterer $i$, respectively, whose general expression is given as
\begin{equation}
    \va_q(\omega) = [1\; e^{j\omega n_{q,1}}\; \cdots\; e^{j\omega n_{q,N_q-1}}]^T, q\in\{\mathrm{t},\mathrm{r}\},\label{eq:a_q_w}
\end{equation}
$\omega_{q,i} = 2\pi d\sin\theta_{q,i}/\lambda_c$, and $\lambda_c$ is the wavelength. Note that $d$ is the antenna unit spacing that is given by $d=\lambda_c/2$ in this paper. In addition, to allow the expression of an arbitrary antenna array architecture, locations of antenna elements in the antenna array are denoted as
$\{0, n_{q,1}d, \dots, n_{q,N_{q}-1}d\}$, $q\in\{\mathrm{t},\mathrm{r}\}$,
% $\{0, n_{\mathrm{t},1}d, \dots, n_{\mathrm{t},N_{\mathrm{t}}-1}d\}$ and $\{0, n_{\mathrm{r},1}d, \dots, n_{\mathrm{r},N_{\mathrm{r}}-1}d\}$
for transmit and receive antenna arrays, respectively.

\begin{figure}[!t]
\centering
\includegraphics[width=2.8in]{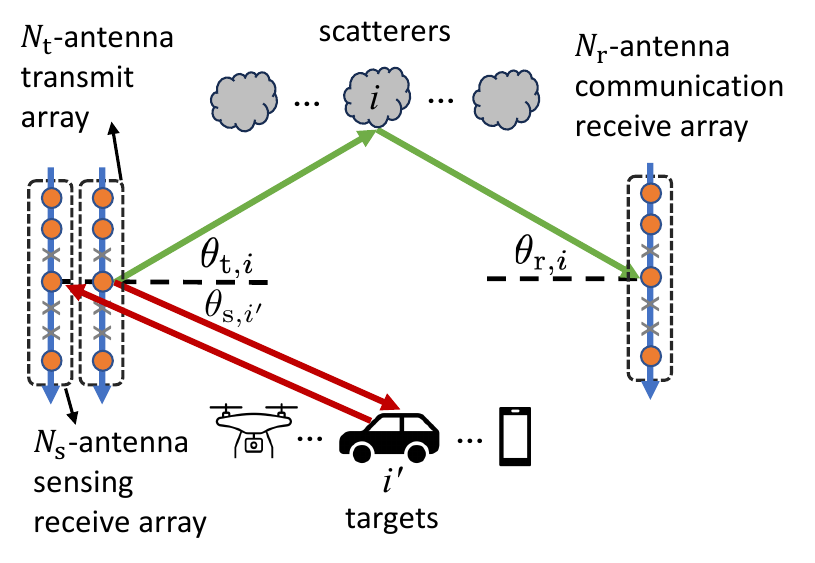}
\vspace{-3mm}
\caption{MIMO ISAC channel model.}
\label{fig:MIMO_channel}
\end{figure}

\subsection{Sensing Signal Model}

The received sensing signal for $K$ independent snapshots within a CPI of a MIMO monostatic sensing system is
\begin{equation}
    \mathbf{y}_{\mathrm{s}}[k] = \mathbf{H}_{\mathrm{s}}[k] \mathbf{x}[k] + \mathbf{w}_{\mathrm{s}}[k], \label{eq:y_s_k}
\end{equation}
for $k = 1, 2, \dots, K$, where
\begin{equation}
    \mathbf{H}_{\mathrm{s}}[k] = \sum_{i=1}^{D_{\mathrm{s}}} \alpha_{\mathrm{s},i}[k] \mathbf{a}_{\mathrm{s}}(\omega_{\mathrm{s},i})\mathbf{a}_{\mathrm{t}}^H(\omega_{\mathrm{s},i})
\end{equation}
is the MIMO sensing channel matrix for snapshot $k$, $\alpha_{\mathrm{s},i}[k]$ is the channel coefficient for the sensing path of target $i$, $\mathbf{w}_{\mathrm{s}}[k]$ is the AWGN at the sensing receiver, $\omega_{\mathrm{s},i} = 2\pi d\sin\theta_{\mathrm{s},i}/\lambda_c$, and $\mathbf{a}_{\mathrm{s}}(\omega_{\mathrm{s},i})$ follows again the expression in \eqref{eq:a_q_w} in consideration of the received sensing antenna array architecture.
Finally, we assume that the path gains $\alpha_{z,i}[k]$ are independent and identically distributed (i.i.d.) over snapshots and uncorrelated with zero mean and variance $\sigma_{\alpha,z}^2$, such that $\E[\alpha_{z,i}[k] \alpha_{z,j}^*[k]] = \sigma_{\alpha,z}^2 \delta_{ij}$, $z \in \{\mathrm{c}, \mathrm{s}\}$.

In the following, we analyze the capacity and angle CRB of the MIMO ISAC system separately, essentially assuming that the communication and sensing channels are independent. Investigating the case where the communication and sensing channels are correlated is regarded as an important future direction.

In this paper, we assume that the antenna locations are all integer multiples of half-wavelength so that $n_{q,i}$ in (\ref{eq:a_q_w}) are integers for all $i$. Then, an $N$-antenna ULA is defined by the antenna locations $\{0, 1, \dots, N - 1\}$, and a SA is, formally speaking, any array with an aperture larger than that of the ULA. However, for typical SAs we consider, the aperture scales as $O(N^2)$, as shown below.
\begin{define}[Nested array \cite{PiyaPPVnested} and coprime array \cite{PPVCoprimePaper}] \label{def:SA}
The antenna locations of a nested array (NA) are $\{0, 1, \dots, N/2 - 1\} \cup \{k(N/2 + 1) - 1 \mid k = 1, 2, \dots, N/2\}$, where $N$ is assumed an even number.
The antenna locations of a coprime array (CA) are $\{k M_{\mathrm{co}} \mid k = 0, 1, \dots, N_{\mathrm{co}} - 1\} \cup \{k N_{\mathrm{co}} \mid k = 1, 2, \dots, M_{\mathrm{co}}\}$, where $M_{\mathrm{co}}$ and $N_{\mathrm{co}} = N - M_{\mathrm{co}}$ are assumed coprime and in $O(N)$.
\end{define}
% \begin{IEEEeqnarray}{rCl}
%     n_{q,i} &=& \begin{cases}
%         i, & 0 \leq i \leq \frac{N_q}{2} - 1 \\
%         (i - \frac{N_q}{2} + 1)(\frac{N_q}{2} + 1) - 1, & \frac{N_q}{2} \leq i \leq N_q - 1.
%     \end{cases} \IEEEeqnarraynumspace
% \end{IEEEeqnarray}
% \begin{IEEEeqnarray}{rCl}
%     \left\{0, 1, \dots, \frac{N}{2} - 1\right\} \cup \left\{k(\frac{N}{2} + 1) - 1 \mid k = 1, 2, \dots, \frac{N}{2}\right\} \IEEEeqnarraynumspace
% \end{IEEEeqnarray}

\section{Proposed Majorization Methodology for Channel Capacity} \label{sec:capacity}

In this section, we develop a general analytical framework for comparing the ergodic capacity of different array architectures by leveraging tools from stochastic majorization theory and Ky Fan norms. This framework enables a principled comparison of the achievable ergodic channel capacities between SAs and ULAs, thereby revealing fundamental performance differences between the two array configurations.
% through the framework of stochastic majorization, moving beyond discrete scenario-based approximations.

\subsection{MIMO Channel Capacity Function}
We consider the MIMO communication channel \eqref{eq:H_c_k}.
Let $\mathbf{G} = \mathbf{H}_{\mathrm{c}}\mathbf{H}_{\mathrm{c}}^H$ be the Gram matrix, and let $\boldsymbol{\lambda} = [\lambda_1\; \cdots\; \lambda_{D_{\mathrm{c}}}]^T$ collect its eigenvalues. Here, we drop the snapshot index $k$ since $\alpha_{\mathrm{c},i}[k]$ and thus $\mathbf{H}_{\mathrm{c}}[k]$ are i.i.d. over snapshots.
Under equal power allocation, the capacity corresponding to $\mathbf{H}_{\mathrm{c}}$ is then characterized by $\boldsymbol{\lambda}$ and is given by
\begin{equation}
    C(\boldsymbol{\lambda}) = \sum_{i=1}^{D_{\mathrm{c}}} \log_2\left(1 + \rho \lambda_i(\mathbf{G})\right)
    \label{eq:capacity_def}
\end{equation}
where $\rho = P / (N_{\mathrm{t}} \sigma^2) > 0$ denotes the per-subchannel SNR. %We note that the function %$\Phi(\boldsymbol{\lambda}) = $\Phi(\boldsymbol{\lambda}) :=\sum_i \log_2(1 + \rho \lambda_i)$ is strictly \textit{Schur-concave} due to the concavity of the logarithm. 
The ergodic capacity is given by $\E[C(\boldsymbol{\lambda})]$.

\subsection{Majorization Theory and Ky Fan \texorpdfstring{$k$}{k}-Norm}
We now briefly review majorization theory, which will be useful in the subsequent analysis.

% --- Definition of Majorization ---
\begin{define}[Majorization \cite{MarshallMajorization}] \label{def:major}
For two vectors $\mathbf{x}, \mathbf{y} \in \mathbb{R}^n$ with components sorted in non-increasing order (i.e., $x_1 \geq x_2 \geq \dots \geq x_n$ and $y_1 \geq y_2 \geq \dots \geq y_n$), we say that $\mathbf{x}$ is majorized by $\mathbf{y}$, denoted as $\mathbf{x} \prec \mathbf{y}$, if:
\begin{enumerate}
    \item $\sum_{i=1}^{k} x_i \leq \sum_{i=1}^{k} y_i$ for $k = 1, \dots, n-1$,
    \item $\sum_{i=1}^{n} x_i = \sum_{i=1}^{n} y_i$.
\end{enumerate}
\end{define}

% --- Statement of Karamata's Inequality ---
For two vectors such that one is majorized by the other, the following Karamata's inequality holds for any convex (or concave) function.
\begin{theorem}[Karamata's inequality \cite{MarshallMajorization}]\label{thm:Kara}
Let $f: I \to \mathbb{R}$ be a convex function on an interval $I \subseteq \mathbb{R}$. If $\mathbf{x}, \mathbf{y} \in I^n$ such that $\mathbf{x} \prec \mathbf{y}$, then:
\begin{equation}
    \sum_{i=1}^{n} f(x_i) \leq \sum_{i=1}^{n} f(y_i).
\end{equation}
Conversely, if $f$ is a concave function (such as the $\log_2(1 + \rho x)$ function used in channel capacity), then:
\begin{equation}
    \sum_{i=1}^{n} f(x_i) \geq \sum_{i=1}^{n} f(y_i).
\end{equation}
\end{theorem}

%\subsection{Ky Fan $k$-Norm}
We then review the Ky Fan $k$-norm for the Hermitian Gram matrix $\mathbf{G}$, which is defined as the sum of the $k$ largest eigenvalues of $\mathbf{G}$.

\begin{define}[Ky Fan $k$-norm \cite{HornTopics}]
Let $\lambda_1(\mathbf{G}) \ge \lambda_2(\mathbf{G}) \ge \dots \ge \lambda_{N_{\mathrm{r}}}(\mathbf{G})$ denote the eigenvalues of the Hermitian Gram matrix $\mathbf{G} = \mathbf{H}\mathbf{H}^H \in \mathbb{C}^{N_{\mathrm{r}} \times N_{\mathrm{r}}}$. The Ky Fan $k$-norm of $\mathbf{G}$, denoted by $\|\mathbf{G}\|_{(k)}$, is defined as the sum of the $k$ largest eigenvalues, i.e.,
\begin{equation}
    \|\mathbf{G}\|_{(k)} = \sum_{i=1}^{k} \lambda_i(\mathbf{G})
    = \max_{\mathbf{V}^H \mathbf{V} = \mathbf{I}_k} 
    \operatorname{tr}(\mathbf{V}^H \mathbf{G} \mathbf{V}),
\end{equation}
where $\mathbf{V} \in \mathbb{C}^{N_{\mathrm{r}} \times k}$ is a matrix with orthonormal columns.
\end{define}

If we regard the Ky Fan $k$-norm as a function of the matrix $\mathbf{G}$, i.e., $f_k(\mathbf{G}) := \|\mathbf{G}\|_{(k)}$, then $f_k(\mathbf{G})$ is convex in $\mathbf{G}$. Therefore, by Jensen's inequality, we obtain
\begin{equation}
    \E[f_k(\mathbf{G})] 
    \ge f_k\!\left(\E[\mathbf{G}]\right).
    \label{eq:f_k_Jensen}
\end{equation}

%\subsubsection{Expected Gram Matrix for Restricted Angular Range}

To further develop the analytical tools, we analyze the expected Gram matrix over a restricted angular range $(-\omega_{\max}, \omega_{\max}]$.  
Consider the MIMO communication channel model in \eqref{eq:H_c_k}, whose
\begin{comment}
\begin{equation}
\mathbf{H} = \sum_{i=1}^{D_{\mathrm{c}}} \alpha_{\mathrm{c},i} \mathbf{a}_{\mathrm{r}}(\omega_{\mathrm{r},i}) \mathbf{a}_{\mathrm{t}}^H(\omega_{\mathrm{t},i})
\end{equation}
\end{comment}
Gram matrix $\mathbf{G} = \mathbf{H}_{\mathrm{c}}\mathbf{H}_{\mathrm{c}}^H$ is
\begin{IEEEeqnarray}{rCl}
\mathbf{G} &=& \sum_{i=1}^{D_{\mathrm{c}}} \alpha_{\mathrm{c},i} \mathbf{a}_{\mathrm{r}}(\omega_{\mathrm{r},i}) \mathbf{a}_{\mathrm{t}}^H(\omega_{\mathrm{t},i}) \sum_{j=1}^{D_{\mathrm{c}}} \alpha_{\mathrm{c},j}^* \mathbf{a}_{\mathrm{t}}(\omega_{\mathrm{t},j}) \mathbf{a}_{\mathrm{r}}^H(\omega_{\mathrm{r},j}). \IEEEeqnarraynumspace
\end{IEEEeqnarray}
We take the expectation over the channel paths, assuming they are uncorrelated with zero mean and variance $\sigma_{\alpha,\mathrm{c}}^2$, such that $\E[\alpha_{\mathrm{c},i} \alpha_{\mathrm{c},j}^*] = \sigma_{\alpha,\mathrm{c}}^2 \delta_{ij}$, which results in
\begin{subequations}
\label{eq:E_G}
\begin{align}%{rCl}
\E[\mathbf{G}] &= \sum_{i=1}^{D_{\mathrm{c}}} \sigma_{\alpha,\mathrm{c}}^2 \E_{\omega_{\mathrm{r},i}} \left[ \mathbf{a}_{\mathrm{r}}(\omega_{\mathrm{r},i}) \mathbf{a}_{\mathrm{t}}^H(\omega_{\mathrm{t},i}) \mathbf{a}_{\mathrm{t}}(\omega_{\mathrm{t},i}) \mathbf{a}_{\mathrm{r}}^H(\omega_{\mathrm{r},i}) \right] \nonumber \\ %\IEEEeqnarraynumspace \label{eq:E_G}
&=
\sum_{i=1}^{D_{\mathrm{c}}} \sigma_{\alpha,\mathrm{c}}^2 N_{\mathrm{t}} \E_{\omega_{\mathrm{r},i}} \left[ \mathbf{a}_{\mathrm{r}}(\omega_{\mathrm{r},i}) \mathbf{a}_{\mathrm{r}}^H(\omega_{\mathrm{r},i}) \right] \label{eq:E_G_a} \\
&=
D_{\mathrm{c}} \sigma_{\alpha,\mathrm{c}}^2 N_{\mathrm{t}} \E_{\omega} \left[ \mathbf{a}_{\mathrm{r}}(\omega) \mathbf{a}_{\mathrm{r}}^H(\omega) \right] \label{eq:E_G_b} \\
&=
D_{\mathrm{c}} \sigma_{\alpha,\mathrm{c}}^2 N_{\mathrm{t}} \mathbf{R}_{\mathrm{r}}, \label{eq:E_G_R}
\end{align}
\end{subequations}
where \eqref{eq:E_G_a} uses the fact that $\mathbf{a}_{\mathrm{t}}^H \mathbf{a}_{\mathrm{t}} = \|\mathbf{a}_{\mathrm{t}}\|^2 = N_{\mathrm{t}}$, \eqref{eq:E_G_b} assumes the spatial frequency $\omega_{\mathrm{r},i}$ follows the same distribution as $\omega$ for all paths $i$, and \eqref{eq:E_G_R} adopts the notation $\mathbf{R}_{\mathrm{r}} := \E_{\omega} \left[ \mathbf{a}_{\mathrm{r}}(\omega) \mathbf{a}_{\mathrm{r}}^H(\omega) \right]$.

%%%%
Let $f_{\Omega}(\omega)$ be the probability density function of the spatial frequency $\omega$. The $(k, m)$-th entry of $\mathbf{R}_{\mathrm{r}}$ is given by
% \begin{align}
%     R_{\mathrm{r},k,m} &= \E[e^{j\omega (n_{\mathrm{r},k} - n_{\mathrm{r},m})}] \nonumber \\%\label{eq:R_rkm} \\
%     &= \int_{-\pi}^{\pi} f_{\Omega}(\omega) e^{j\omega (n_{\mathrm{r},k} - n_{\mathrm{r},m})} d\omega \nonumber \\
%     &= 2\pi a_{n_{\mathrm{r},m} - n_{\mathrm{r},k}}, \label{eq:R_rkm}
% \end{align}
\begin{IEEEeqnarray}{rCl}
    R_{\mathrm{r},k,m} &=& \E[e^{j\omega (n_{\mathrm{r},k} - n_{\mathrm{r},m})}]
    = \int_{-\pi}^{\pi} f_{\Omega}(\omega) e^{j\omega (n_{\mathrm{r},k} - n_{\mathrm{r},m})} d\omega \IEEEnonumber \\
    &=& 2\pi a_{n_{\mathrm{r},m} - n_{\mathrm{r},k}}, \label{eq:R_rkm}
\end{IEEEeqnarray}
where $a_l$'s are the Fourier series coefficients of $f_{\Omega}(\omega)$. As established in classical Fourier analysis \cite{GrafakosFourier}, the decay rate of the Fourier coefficients $a_l$ is governed by the smoothness of the probability density function $f_{\Omega}(\omega)$. Specifically, if the $p$-th derivative of $f_{\Omega}(\omega)$ exists and is of bounded variation, $a_l$ decay as $O(|l|^{-(p+1)})$. As a special case with $p=0$, if $f_{\Omega}(\omega)$ is of bounded variation, $a_l$ decays as $O(|l|^{-1})$. Hence, under almost all practical channel conditions, represented by the distribution $f_{\Omega}(\omega)$ of the path angles, the off-diagonal entries $R_{\mathrm{r},k,m}$ diminishes as the antenna spacings $n_{\mathrm{r},m} - n_{\mathrm{r},k}$ increase.%, so the majorization criterion (\ref{eq:mjrz_criterion}) is likely to be satisfied.

We note that computing $R_{\mathrm{r},k,m}$ does not require Monte Carlo simulation if closed-form expressions for the Fourier series coefficients are available, or if the coefficients are evaluated using numerical integration methods like Riemann sum approximation. In what follows, we present two examples.

\begin{example}\label{ex:unif}
Consider that the spatial frequency $\omega$ is uniformly distributed over the interval $(-\omega_{\text{max}}, \omega_{\text{max}}]$, where $0 < \omega_{\text{max}} \leq \pi$, with probability density function
\begin{equation}
f_{\Omega}(\omega) =
\begin{cases}
\frac{1}{2\omega_{\text{max}}}, & |\omega| \leq \omega_{\text{max}},\\
0, & \text{otherwise}.
\end{cases}
\label{eq:PDF_omega}
\end{equation}
This assumption reflects the clustered scattering characteristic commonly observed in mmWave MIMO channels, where propagation paths tend to concentrate within a limited angular region around a dominant direction~\cite{Heath}. In this model, $\omega_{\text{max}}$ characterizes the angular spread of the cluster, and a smaller $\omega_{\text{max}}$ corresponds to a channel with more tightly concentrated paths in the angle domain. %Recall that $\mathbf{a}_{\mathrm{r}}(\omega) = [e^{j\omega n_{\mathrm{r},0}}, e^{j\omega n_{\mathrm{r},1}}, \dots, e^{j\omega n_{\mathrm{r},N_{\mathrm{r}}-1}}]^T$. 
Under the above assumption on $\omega$, we evaluate the $(k,m)$-th entry of the matrix $\E_{\omega}\!\left[\mathbf{a}_{\mathrm{r}}(\omega)\mathbf{a}_{\mathrm{r}}^H(\omega)\right]$, denoted by $R_{\mathrm{r},k,m}$, as follows:
\begin{IEEEeqnarray}{rCl}
    R_{\mathrm{r},k,m} &=& \E[e^{j\omega n_{\mathrm{r},k}} (e^{j\omega n_{\mathrm{r},m}})^*] = \E[e^{j\omega (n_{\mathrm{r},k} - n_{\mathrm{r},m})}] \IEEEnonumber \\
    &=&
    \frac{1}{2\omega_{\text{max}}} \int_{-\omega_{\text{max}}}^{\omega_{\text{max}}} e^{j\omega (n_{\mathrm{r},k} - n_{\mathrm{r},m})} d\omega. \label{eq:Rkm_pre}
\end{IEEEeqnarray}
Let $d_{\mathrm{r},k,m} = n_{\mathrm{r},k} - n_{\mathrm{r},m}$. If $k = m$, then $d_{\mathrm{r},k,m} = 0$ and $R_{\mathrm{r},k,k} = 1$. If $k \neq m$, we have
\begin{align}
    \eqref{eq:Rkm_pre}
    &=
    \frac{1}{2\omega_{\text{max}}} \left[ \frac{e^{j\omega d_{\mathrm{r},k,m}}}{j d_{\mathrm{r},k,m}} \right]_{-\omega_{\text{max}}}^{\omega_{\text{max}}} \nonumber \\
    &= 
    \frac{e^{j\omega_{\text{max}} d_{\mathrm{r},k,m}} - e^{-j\omega_{\text{max}} d_{\mathrm{r},k,m}}}{2j \omega_{\text{max}} d_{\mathrm{r},k,m}} =
    \frac{\sin(\omega_{\text{max}} d_{\mathrm{r},k,m})}{\omega_{\text{max}} d_{\mathrm{r},k,m}} \nonumber \\
    &= \text{sinc}(\omega_{\text{max}} (n_{\mathrm{r},k} - n_{\mathrm{r},m})/\pi), \label{eq:Rkm}
\end{align}
where we use Euler's identity $\sin(x) = \frac{e^{jx} - e^{-jx}}{2j}$ and the definition of sinc function. Hence, $\mathbf{R}_{\mathrm{r}}$ (and thus $\E[\mathbf{G}]$) is completely characterized.
\end{example}
\begin{example}\label{ex:gaussian}
    Consider a channel with a Gaussian-shaped angular density function $f_{\Omega}(\omega)$. Since $f_{\Omega}(\omega)$ is infinitely differentiable, its Fourier coefficients $a_l$ decay exponentially fast.
\end{example}

Several remarks about $\E[\mathbf{G}]$ are worth noting.
\begin{remark}[Trace conservation]
For any $0<\omega_{\text{max}}\leq \pi$, we have
    \begin{IEEEeqnarray}{rCl}
        \text{tr}(\E[\mathbf{G}]) &=& f_{N_{\mathrm{r}}}(\E[\mathbf{G}]) = \sum_{i=1}^{N_{\mathrm{r}}} \lambda_i(\E[\mathbf{G}]) = D_{\mathrm{c}} \sigma_{\alpha,\mathrm{c}}^2 N_{\mathrm{t}} N_{\mathrm{r}} \IEEEeqnarraynumspace \label{eq:trace_E_G}
    \end{IEEEeqnarray}
    for any array.
\end{remark}

\begin{remark}[Whole-range uniform distribution of angles]
    According to Example~\ref{ex:unif}, if $\omega_{\text{max}}=\pi$, \eqref{eq:Rkm} becomes sinc$(n_{\mathrm{r},k} - n_{\mathrm{r},m})$ and we obtain
    \begin{equation}
        \E[\mathbf{G}] = (D_{\mathrm{c}} \sigma_{\alpha,\mathrm{c}}^2 N_{\mathrm{t}}) \mathbf{I}_{N_{\mathrm{r}}} \label{eq:whole_unif}
    \end{equation}
    for any array.
\end{remark}

\subsection{Proposed Analytical Framework: SA vs. ULA} \label{ssc:ULA_mjrz_SA}
In this subsection, we present the proposed analytical framework for comparing a SA and a ULA based on the properties of their respective Gram matrices. Let $\mathbf{G}_{\mathrm{SA}}$ and $\mathbf{G}_{\mathrm{ULA}}$ denote the Gram matrices corresponding to a SA and a ULA, respectively, and let $\boldsymbol{\lambda}_{\mathrm{SA}}$ and $\boldsymbol{\lambda}_{\mathrm{ULA}}$ denote their eigenvalue vectors. If one could establish that
$\boldsymbol{\lambda}_{\mathrm{SA}} \prec \boldsymbol{\lambda}_{\mathrm{ULA}}$ for all realizations of the channel,
then, by Karamata's inequality in Theorem~\ref{thm:Kara} together with the concavity of the capacity function, it would follow that
$C(\boldsymbol{\lambda}_{\mathrm{SA}}) \ge C(\boldsymbol{\lambda}_{\mathrm{ULA}}),$
and hence
$\E\!\left[C(\boldsymbol{\lambda}_{\mathrm{SA}})\right] \ge \E\!\left[C(\boldsymbol{\lambda}_{\mathrm{ULA}})\right].$

However, proving
%\[
$\boldsymbol{\lambda}_{\mathrm{SA}} \prec \boldsymbol{\lambda}_{\mathrm{ULA}}$
%\]
is generally very difficult, and in many cases it is in fact false, since majorization is a rather stringent ordering. Motivated by this observation, we instead consider the following more tractable surrogate criterion:
\begin{framed}
\begin{equation}
\boldsymbol{\lambda}\!\bigl(\E[\mathbf{G}_{\mathrm{SA}}]\bigr)
\prec
\boldsymbol{\lambda}\!\bigl(\E[\mathbf{G}_{\mathrm{ULA}}]\bigr).
\label{eq:mjrz_criterion_G}
\end{equation}
\end{framed}
We note that, given \eqref{eq:E_G}, the criterion \eqref{eq:mjrz_criterion_G} is equivalent to
\begin{framed}
    \begin{equation}
    \boldsymbol{\lambda}(\mathbf{R}_{\text{r,SA}}) \prec \boldsymbol{\lambda}(\mathbf{R}_{\text{r,ULA}}). \label{eq:mjrz_criterion}
    \end{equation}
\end{framed}

In the sequel, we provide the rationale for adopting this criterion.
\begin{enumerate}
    \item If we have $\boldsymbol{\lambda}(\E[\mathbf{G}_{\text{SA}}]) \prec \boldsymbol{\lambda}(\E[\mathbf{G}_{\text{ULA}}])$ in \eqref{eq:mjrz_criterion_G}, then it is \textbf{likely} that
    \begin{equation}
        \E[\boldsymbol{\lambda}(\mathbf{G}_{\mathrm{SA}})] \prec \E[\boldsymbol{\lambda}(\mathbf{G}_{\mathrm{ULA}})],
        \label{eq:avg_mjrz}
    \end{equation}
    provided that the Jensen gaps of SA and ULA in
    \begin{equation}
        \E[f_k(\mathbf{G}_{\mathrm{SA}})] \geq f_k(\E[\mathbf{G}_{\mathrm{SA}}]),
    \end{equation}
    and
    \begin{equation}
        \E[f_k(\mathbf{G}_{\mathrm{ULA}})] \geq f_k(\E[\mathbf{G}_{\mathrm{ULA}}]), 
    \end{equation}
    $\forall k=1, \dots, N_{\mathrm{r}}-1$, do not differ too much.
    \item If we have $\E[\boldsymbol{\lambda}(\mathbf{G}_{\mathrm{SA}})] \prec \E[\boldsymbol{\lambda}(\mathbf{G}_{\mathrm{ULA}})]$ in \eqref{eq:avg_mjrz}, then since $\Phi(\boldsymbol{\lambda}) := \sum_i \log_2(1 + \rho \lambda_i)$ is composed of the concave function $\log_2(1 + \rho x)$, we can use the Karamata's inequality in Theorem~\ref{thm:Kara} to obtain
    \begin{equation}
    \Phi(\E[\boldsymbol{\lambda}_{\mathrm{SA}}]) \geq \Phi(\E[\boldsymbol{\lambda}_{\mathrm{ULA}}]).
    \label{eq:phi_ineq}
    \end{equation}
%Let $\bar{\boldsymbol{\lambda}}_{\mathrm{SA}} = \E[\boldsymbol{\lambda}_{\mathrm{SA}}] = \E[\boldsymbol{\lambda}(\mathbf{G}_{\mathrm{SA}})]$ and $\bar{\boldsymbol{\lambda}}_{\mathrm{ULA}} = \E[\boldsymbol{\lambda}_{\mathrm{ULA}}] = \E[\boldsymbol{\lambda}(\mathbf{G}_{\mathrm{ULA}})]$.   
    \item If we have $\Phi(\E[\boldsymbol{\lambda}_{\mathrm{SA}}]) \geq \Phi(\E[\boldsymbol{\lambda}_{\mathrm{ULA}}])$ in \eqref{eq:phi_ineq}, it is \textbf{likely} that we have the following relationship between the ergodic capacities of the SA and ULA,
    \begin{align}
        \E[C(\boldsymbol{\lambda}_{\mathrm{SA}})]
        &= \E[\Phi(\boldsymbol{\lambda}_{\mathrm{SA}})] \nonumber \\
        &\geq \E[\Phi(\boldsymbol{\lambda}_{\mathrm{ULA}})]
        = \E[C(\boldsymbol{\lambda}_{\mathrm{ULA}})], \label{eq:ergCapIneq}
    \end{align}
    provided that the Jensen gaps of SA and ULA in
    \begin{equation}
        \E[\Phi(\boldsymbol{\lambda}_{\mathrm{SA}})] \leq \Phi(\E[\boldsymbol{\lambda}_{\mathrm{SA}}]),
    \end{equation}
    and
    \begin{equation}
         \E[\Phi(\boldsymbol{\lambda}_{\mathrm{ULA}})] \leq \Phi(\E[\boldsymbol{\lambda}_{\mathrm{ULA}}]),
    \end{equation}
    do not differ too much.
\end{enumerate}
Note that we can only conclude that (\ref{eq:ergCapIneq}) is likely (but not definitely) true due to the Jensen gaps. However, simulations in Sec. \ref{sec:sim} show that (\ref{eq:ergCapIneq}) is indeed established in a wide range of channel conditions.

    Let us revisit the examples discussed above. 
    %\addtocounter{example}{-2}
    \begin{example}[Uniform revisited]
        For the case where each DOA $\omega_{\mathrm{r},i}$ is uniformly distributed in $(-\omega_{\max},\omega_{\max}]$, we know from Example~\ref{ex:unif} that
        % $R_{\mathrm{r,SA},k,m} = \mathrm{sinc}\!\left(\frac{\omega_{\max} d_{\mathrm{r,SA},k,m}}{\pi}\right),$ and $R_{\mathrm{r,ULA},k,m} = \mathrm{sinc}\!\left(\frac{\omega_{\max} d_{\mathrm{r,ULA},k,m}}{\pi}\right).$
        $R_{\mathrm{r,SA},k,m} = \mathrm{sinc}(\omega_{\max} d_{\mathrm{r,SA},k,m} / \pi)$ and $R_{\mathrm{r,ULA},k,m} = \mathrm{sinc}(\omega_{\max} d_{\mathrm{r,ULA},k,m} / \pi)$.

    For the special case $\omega_{\max}=\pi$, we have $\boldsymbol{\lambda}(\mathbf{R}_{\mathrm{r,SA}})=
    \boldsymbol{\lambda}(\mathbf{R}_{\mathrm{r,ULA}}),$ which trivially satisfies $\boldsymbol{\lambda} (\mathbf{R}_{\mathrm{r,SA}}) \prec \boldsymbol{\lambda}(\mathbf{R}_{\mathrm{r,ULA}}).$ For the general case $0<\omega_{\max}\le \pi$, recall that a vector $\mathbf{x}$ is majorized by $\mathbf{y}$ ($\mathbf{x}\prec\mathbf{y}$) if $\mathbf{x}$ is ``less spread out'' than $\mathbf{y}$. For the eigenvalue vector $\boldsymbol{\lambda}$ of an $N_{\mathrm r}\times N_{\mathrm r}$ positive semi-definite matrix $\mathbf{R}_{\mathrm r}$ with a fixed trace $\mathrm{tr}(\mathbf{R}_{\mathrm r})=N_{\mathrm r}$, the identity matrix $\mathbf{I}$ (whose eigenvalues satisfy $\lambda_i=1$ for all $i$) corresponds to the lower bound in the majorization order, representing the most uniform eigenvalue distribution.
    In a SA, the antenna indices $n_{\mathrm r,k}$ are chosen such that the differences $|n_{\mathrm r,k}-n_{\mathrm r,m}|$ are generally larger than those in a ULA for $k\neq m$. Since the sinc function typically decreases in magnitude as its argument increases, we have
\begin{align}
|R_{\mathrm{r,SA},k,m}|
&=
\left|
\mathrm{sinc}\!\left(\frac{\omega_{\max} d_{\mathrm{r,SA},k,m}}{\pi}\right)
\right| \nonumber\\
&\le
\left|
\mathrm{sinc}\!\left(\frac{\omega_{\max} d_{\mathrm{r,ULA},k,m}}{\pi}\right)
\right|
=
|R_{\mathrm{r,ULA},k,m}|, \nonumber
\end{align}
for most off-diagonal pairs.

Consequently, the off-diagonal entries of $\mathbf{R}_{\mathrm{r,SA}}$ tend to be smaller in magnitude than those of $\mathbf{R}_{\mathrm{r,ULA}}$. As the off-diagonal entries decrease, the matrix becomes more diagonally dominant, which tends to push its eigenvalues toward the uniform vector $[1\; 1\; \cdots\; 1]^T$. In contrast, the ULA has the smallest aperture for a fixed $N_{\mathrm r}$, resulting in larger off-diagonal terms and hence a more uneven (or ``peaked'') eigenvalue distribution. Therefore, it is reasonable to expect that the majorization condition in~\eqref{eq:mjrz_criterion} holds.
\end{example}

\begin{comment}
    \begin{equation}
        f_k(\E[\mathbf{G}_{\mathrm{SA}}]) \le f_k(\E[\mathbf{G}_{\mathrm{ULA}}]), \quad \forall k=1, \dots, N_{\mathrm{r}}-1
    \end{equation}
    and
    \begin{equation}
        \boldsymbol{\lambda}(\E[\mathbf{G}_{\text{SA}}]) \prec \boldsymbol{\lambda}(\E[\mathbf{G}_{\text{ULA}}])
    \end{equation}
\end{comment}
    
    \begin{example}[Gaussian revisited]
        As stated in Example~\ref{ex:gaussian}, for a Gaussian-shaped angular density $f_{\Omega}(\omega)$, the Fourier coefficients $a_\ell$ decay exponentially fast with $|\ell|$. Since the entries of $\mathbf{R}_{\mathrm r}$ are given by $R_{\mathrm r,k,m}=a_{n_{\mathrm r,k}-n_{\mathrm r,m}}$, larger antenna separations lead to much smaller off-diagonal correlations. Because SAs typically induce larger pairwise separations than a ULA, the off-diagonal entries of $\mathbf{R}_{\mathrm{r,SA}}$ tend to decay more rapidly, which in turn tends to reduce $\|\mathbf{R}_{\mathrm{r,SA}}\|_F^2$ relative to $\|\mathbf{R}_{\mathrm{r,ULA}}\|_F^2$. Under the common trace constraint, this means that the eigenvalue vector of $\mathbf{R}_{\mathrm{r,SA}}$ is expected to be closer to the uniform vector than that of $\mathbf{R}_{\mathrm{r,ULA}}$. Therefore, the majorization condition in \eqref{eq:mjrz_criterion} becomes even more plausible for Gaussian-shaped angular spectra, although a formal proof generally remains difficult.
    \end{example} 

\subsection{Relationship With Effective DoF and Beam Pattern}\label{ssc:edof}
We now connect the proposed analytical framework with the {\it effective degree of freedom (EDoF)}, which is defined as follows.

\begin{define}[EDoF \cite{WangFeng,Muharemovic}]
For a fixed MIMO channel $\mathbf{H}_{\mathrm{c}}$, the EDoF is defined as
\begin{align}
    \epsilon 
    := \frac{\left[ \text{tr}(\mathbf{H}_{\mathrm{c}}\mathbf{H}_{\mathrm{c}}^H) \right]^2}{\|\mathbf{H}_{\mathrm{c}}\mathbf{H}_{\mathrm{c}}^H\|_F^2}  
    = \frac{\left[\text{tr}(\mathbf{G})\right]^2}{\text{tr}(\mathbf{G}^2)}  
    = \frac{\left[\sum_{i=1}^{D_{\mathrm{c}}} \lambda_i(\mathbf{G})\right]^2}{\sum_{i=1}^{D_{\mathrm{c}}} \lambda_i^2(\mathbf{G})}.
\end{align}
%where $\mathbf{G}=\mathbf{H}_{\mathrm{c}}\mathbf{H}_{\mathrm{c}}^H$ and $\lambda_i(\mathbf{G})$ denotes the eigenvalues of $\mathbf{G}$.
\end{define}
Since the channel matrix $\mathbf{H}_{\mathrm{c}}$ is random, the EDoF defined above is also random and therefore not directly suitable as a performance metric for evaluating different array geometries. We therefore consider the following deterministic surrogate obtained by averaging.

\begin{define}
The average surrogate of the EDoF is defined as
\begin{align}
\epsilon_{\mathrm{avg}} 
:=
\frac{\left[\E\!\left[\sum_{i=1}^{D_{\mathrm{c}}} \lambda_i(\mathbf{G})\right]\right]^2}
{\E\!\left[\sum_{i=1}^{D_{\mathrm{c}}} \lambda_i^2(\mathbf{G})\right]}. \label{eq:EDoFavg}
\end{align}
\end{define}

\begin{remark}
Using Cauchy-Schwarz inequality and the fact that each $\lambda_i(\mathbf{G}) \geq 0$, we can obtain $1 \leq \epsilon \leq D_{\mathrm{c}}$.
Moreover, if the eigenvalues are more ``peaked" or spread-out, the sum of squared eigenvalues gets larger, and $\epsilon$ gets smaller.
If the eigenvalues are closer to a uniform vector, the sum of squared eigenvalues gets smaller, and $\epsilon$ gets larger.
In the two extreme cases, $\epsilon = 1$ if there is only one nonzero eigenvalue, and $\epsilon = D_{\mathrm{c}}$ if all the eigenvalues $\lambda_i(\mathbf{G})$, $i = 1, \dots, D_{\mathrm{c}}$ are equal.
\end{remark}

\begin{remark}
    The surrogate $\epsilon_{\mathrm{avg}}$ replaces the ratio of random quantities by the ratio of their expectations to obtain a deterministic and analytically tractable metric. Although $\E[X/Y] \neq \E[X]/\E[Y]$ in general, this approximation is justified when the involved quantities exhibit concentration (e.g., due to channel hardening \cite{Hochwald} or in high-dimensional systems), under which the ratio becomes tightly concentrated around the ratio of expectations. Furthermore, $\epsilon_{\mathrm{avg}}$ depends only on the first- and second-order spectral moments of $\mathbf{G}$, thereby retaining the essential information about the eigenvalue spread that governs the effective degrees of freedom.
\end{remark}

We establish the following closed-form expression for $\epsilon_{\mathrm{avg}}$.

\begin{theorem}\label{thm:EDoF}
For the MIMO channel $\mathbf{H}_{\mathrm{c}}$ in \eqref{eq:H_c_k}, the average surrogate EDoF is given by
\begin{equation}
\epsilon_{\mathrm{avg}}
=
\frac{D_{\mathrm{c}}}
{2 + (D_{\mathrm{c}} - 1)
\left[
\frac{\|\mathbf{R}_{\mathrm{r}}\|_F^2}{N_{\mathrm{r}}^2}
+
\frac{\|\mathbf{R}_{\mathrm{t}}\|_F^2}{N_{\mathrm{t}}^2}
\right]} \label{eq:EDoF_avg}
\end{equation}
and $\|\mathbf{R}_q\|_F^2 = \E [ | \mathbf{a}_q^H(\omega_{q,i}) \mathbf{a}_q(\omega_{q,j}) |^2 ]$, $q \in \{\mathrm{t},\mathrm{r}\}$.
\end{theorem}

\begin{IEEEproof}
See Appendix \ref{adx:avgEDoF}.
\end{IEEEproof}

From Theorem~\ref{thm:EDoF}, we observe that maximizing $\epsilon_{\mathrm{avg}}$ amounts to minimizing $\|\mathbf{R}_q\|_F^2 = \E [ | \mathbf{a}_q^H(\omega_{q,i}) \mathbf{a}_q(\omega_{q,j}) |^2 ]$, the expected values of the \textit{beam patterns} of the arrays. This observation is closely related to the majorization criterion in~\eqref{eq:mjrz_criterion}. As discussed in Sec.~\ref{ssc:ULA_mjrz_SA}, the condition in~\eqref{eq:mjrz_criterion} is likely to hold because the off-diagonal entries of $\mathbf{R}_{\mathrm{r,SA}}$ are generally smaller than those of $\mathbf{R}_{\mathrm{r,ULA}}$, while both matrices have the same trace. Smaller off-diagonal entries lead to a smaller Frobenius norm $\|\mathbf{R}_{\mathrm{r}}\|_F^2$. Consequently, $\|\mathbf{R}_{\mathrm{r}}\|_F^2$, or equivalently the surrogate EDoF $\epsilon_{\mathrm{avg}}$, can serve as a quantifiable proxy for the majorization criterion.

While the majorization condition in~\eqref{eq:mjrz_criterion} provides the theoretical ``on-off'' condition for capacity improvement, the quantity $\|\mathbf{R}_{\mathrm{r}}\|_F^2$ offers a practical metric that enables quantitative comparison between different array geometries. In particular, comparing $\|\mathbf{R}_{\mathrm{r}}\|_F^2$ between a SA and a ULA allows us to assess how well the condition in~\eqref{eq:mjrz_criterion} is satisfied and to estimate the potential capacity gain.

\section{Proposed Angle CRB Analysis} \label{sec:CRB}

In this section, we analyze the angle CRB for the MIMO ISAC system while considering the array architecture. To this end, we will first provide the standard DOA estimation framework for the MIMO ISAC system. Then, the asymptotic angle CRB will be derived.

\subsection{Standard DOA Estimation Framework for the MIMO ISAC System}

The received sensing signal \eqref{eq:y_s_k} can be rewritten as
\begin{IEEEeqnarray}{rCl}
    \mathbf{y}_{\mathrm{s}}[k] &=& \sum_{i=1}^{D_{\mathrm{s}}} \mathbf{a}_{\mathrm{s}}(\omega_{\mathrm{s},i}) s_i[k] + \mathbf{w}_{\mathrm{s}}[k]
    = \mathbf{A}_{\mathrm{s}}(\boldsymbol{\omega}_{\mathrm{s}}) \mathbf{s}[k] + \mathbf{w}_{\mathrm{s}}[k], \IEEEeqnarraynumspace
    \label{eq:sensing_model_DOA}    
\end{IEEEeqnarray}
where $s_i[k] = \alpha_{\mathrm{s},i}[k] \mathbf{a}_{\mathrm{t}}^H(\omega_{\mathrm{s},i}) \mathbf{x}[k]$
% \begin{equation}
%     s_i[k] = \alpha_{\mathrm{s},i}[k] \mathbf{a}_{\mathrm{t}}^H(\omega_{\mathrm{s},i}) \mathbf{x}[k]
% \end{equation}
represents the complex amplitude of the $i$-th target at the $k$-th snapshot, $\mathbf{s}[k] = [s_1[k]\; \cdots\; s_{D_{\mathrm{s}}}[k]]^T$,
% is the corresponding signal vector
and
$\mathbf{A}_{\mathrm{s}}(\boldsymbol{\omega}_{\mathrm{s}}) = [\mathbf{a}_{\mathrm{s}}(\omega_{\mathrm{s},1})\; \cdots\; \mathbf{a}_{\mathrm{s}}(\omega_{\mathrm{s},D_{\mathrm{s}}})]$ is the array manifold matrix.
The sensing signal model in \eqref{eq:sensing_model_DOA} follows the standard DOA estimation framework (e.g., (2.1) of \cite{StoicaStochasticCRB}), where $s_i[k]$ is viewed as the point source at snapshot $k$, but one major difference is that unlike in \cite{StoicaStochasticCRB}, $\mathbf{s}[k]$ here depends on the angle parameters $\boldsymbol{\omega}_{\mathrm{s}} = [\omega_1\; \cdots\; \omega_{D_\mathrm{s}}]^T$.

Subsequently, by using the \textit{conditional model} for CRB from \cite{LiuCRB,StoicaStochasticCRB}, the derivation of angle CRB conditioning on the transmitted symbols $\mathbf{x}[k]$ is in general related to finding the empirical source covariance matrix $\hat{\mathbf{R}}_{\mathrm{s}}$ defined as
\begin{equation}
    \hat{\mathbf{R}}_{\mathrm{s}} = \frac{1}{K} \sum\nolimits_{k=1}^{K} \mathbf{s}[k]\mathbf{s}^H[k]\in \mathbb{C}^{D_{\mathrm{s}} \times D_{\mathrm{s}}}. \label{eq:R_s}
\end{equation}
The $(i, j)$-th entry of $\hat{\mathbf{R}}_{\mathrm{s}}$ represents the correlation between the reflected signals of target $i$ and target $j$ and is
\begin{IEEEeqnarray}{rCl}
    [\hat{\mathbf{R}}_{\mathrm{s}}]_{i,j} &=& \mathbf{a}_{\mathrm{t}}^H(\omega_{\mathrm{s},i}) \left[ \frac{1}{K} \sum_{k=1}^{K} \alpha_{\mathrm{s},i}[k] \alpha_{\mathrm{s},j}^*[k] \mathbf{x}[k] \mathbf{x}^H[k] \right] \mathbf{a}_{\mathrm{t}}(\omega_{\mathrm{s},j}). \IEEEnonumber
\end{IEEEeqnarray}
We assume that the transmitted symbols $\mathbf{x}[k]$ are realizations of i.i.d. random variables with zero mean and $E[\mathbf{x}[k] \mathbf{x}^H[k]] = \frac{P}{N_{\mathrm{t}}} \mathbf{I}_{N_{\mathrm{t}}}$.\footnote{In this case, system basically adopts the simple equal-power transmission, aligned with the consideration of the capacity expression in \eqref{eq:capacity_def}.} By following the law of large numbers with $K \to \infty$, the sample mean converges to the expected value,\footnote{Indeed, when the communication symbols are drawn from a balanced constellation (like QAM or PSK), it converges quickly to the expected value.} leading to
\begin{IEEEeqnarray}{rCl}
     [\hat{\mathbf{R}}_{\mathrm{s}}]_{i,j} &\stackrel{K \to \infty}{=}& \mathbf{a}_{\mathrm{t}}^H(\omega_{\mathrm{s},i}) \E\left[ \alpha_{\mathrm{s},i}[k] \alpha_{\mathrm{s},j}^*[k] \mathbf{x}[k] \mathbf{x}^H[k] \right] \mathbf{a}_{\mathrm{t}}(\omega_{\mathrm{s},j}) \IEEEnonumber \\
     &=& \sigma_{\alpha,\mathrm{s}}^2 \delta_{ij} P.
\end{IEEEeqnarray}
As a result, we can conclude that for large $K$,
\begin{equation}
    \hat{\mathbf{R}}_{\mathrm{s}} \approx \sigma_{\alpha,\mathrm{s}}^2 P \mathbf{I}_{D_{\mathrm{s}}}, \label{eq:R_s_approx}
\end{equation}
where the approximation becomes exact almost surely as $K\to\infty$.
Simulations in Sec. \ref{sec:sim} show that (\ref{eq:R_s_approx}) is highly accurate even for a moderate number of snapshots, such as $K=32$.

\subsection{Asymptotic Analysis for Angle CRB}

By casting the received sensing signal into the standard DOA estimation framework \eqref{eq:sensing_model_DOA}, we can derive the \textit{conditional CRB} for the $i$-th angle estimate $\omega_{\mathrm{s},i}$ with a procedure similar to those in \cite{StoicaStochasticCRB,LiuCRB}.
Namely, we derive the Fisher information matrix for the unknown parameters $[\bm{\omega}_{\mathrm{s}}^T\; \text{Re}\{\bm{\alpha}^T_{\mathrm{s}}[k]\}\; \text{Im}\{\bm{\alpha}^T_{\mathrm{s}}[k]\}]^T$, where $\bm{\alpha}_{\mathrm{s}}[k] = [\alpha_{\mathrm{s},1}[k]\; \cdots\; \alpha_{\mathrm{s},D_\mathrm{s}}[k]]^T$, and take the Schur complement of the block corresponding to $\bm{\omega}_{\mathrm{s}}$.
Then, by the product rule for derivatives, we obtain the conditional CRB for the DOA parameters
\begin{equation}
    \text{CRB}(\boldsymbol{\omega}_{\mathrm{s}}) = \frac{\sigma^2}{2K} \left[ \text{Re} \left\{ \frac{1}{K}\sum_{k=1}^{K} \mathbf{C}^H[k] \mathbf{P}_{\mathbf{A}_{\mathrm{s}}}^{\perp} \mathbf{C}[k] \right\} \right]^{-1},
\end{equation}
where $\sigma^2$ is the noise power, $\mathbf{P}_{\mathbf{A}_{\mathrm{s}}}^{\perp} = \mathbf{I} - \mathbf{A}_{\mathrm{s}}(\mathbf{A}_{\mathrm{s}}^H \mathbf{A}_{\mathrm{s}})^{-1} \mathbf{A}_{\mathrm{s}}^H$ is the orthogonal projection matrix of $\mathbf{A}_{\mathrm{s}}$, $\mathbf{C}[k] = \mathbf{D}_{\mathrm{s}}\mathbf{S}[k] + \mathbf{A}_{\mathrm{s}} \mathbf{S}'[k]$, $\mathbf{S}[k]$ and $\mathbf{S}'[k]$ are diagonal matrices with $[\mathbf{S}[k]]_{ii} = s_{i}[k]$ and $[\mathbf{S}'[k]]_{ii} = \frac{\partial}{\partial \omega_{s,i}} s_{i}[k]$, and the matrix $\mathbf{D}_{\mathrm{s}} \in \mathbb{C}^{N_{\mathrm{s}} \times D_{\mathrm{s}}}$ contains the partial derivatives of the receive steering vectors with respect to the DOA parameters $\boldsymbol{\omega}_{\mathrm{s}}$, given by $\mathbf{D}_{\mathrm{s}} = [\mathbf{d}_{\mathrm{s}}(\omega_{\mathrm{s},1})\; \cdots\; \mathbf{d}_{\mathrm{s}}(\omega_{\mathrm{s},D_{\mathrm{s}}})]$, where $\mathbf{d}_{\mathrm{s}}(\omega_{\mathrm{s},i})$ is the derivative of the $i$-th steering vector, namely, $\mathbf{d}_{\mathrm{s}}(\omega_{\mathrm{s},i}) = \frac{\partial \mathbf{a}_{\mathrm{s}}(\omega)}{\partial \omega} \big|_{\omega = \omega_{\mathrm{s},i}}$.
Then, since $\mathbf{P}_{\mathbf{A}_{\mathrm{s}}}^{\perp} \mathbf{A}_{\mathrm{s}} = \mathbf{0}$, the CRB is simplified as
\begin{equation}
    \text{CRB}(\boldsymbol{\omega}_{\mathrm{s}}) = \frac{\sigma^2}{2K} \left[ \text{Re} \left\{ (\mathbf{D}_{\mathrm{s}}^H \mathbf{P}_{\mathbf{A}_{\mathrm{s}}}^{\perp} \mathbf{D}_{\mathrm{s}}) \circ \hat{\mathbf{R}}_{\mathrm{s}}^T \right\} \right]^{-1}. \label{eq:CRB_exact}
\end{equation}

\begin{remark}
The second term $\mathbf{A}_{\mathrm{s}} \mathbf{S}'[k]$ in $\mathbf{C}[k]$ appears due to the dependence of the signal $\mathbf{s}[k]$ on $\boldsymbol{\omega}_{\mathrm{s}}$.
This dependence leads to just a scalar beamforming gain factor $\mathbf{a}_{\mathrm{t}}^H(\omega_{\mathrm{s},i}) \mathbf{x}[k]$ for each DOA in the received signal $\mathbf{y}_{\mathrm{s}}[k]$.
If the path gain $\alpha_{\mathrm{s},i}[k]$ was known, the CRB could be reduced. However, in practice, the path gain is unknown, so no additional information can be obtained, and the CRB reduces to \eqref{eq:CRB_exact}, which is the same as the case where the signal is independent of $\boldsymbol{\omega}_{\mathrm{s}}$ \cite{LiuCRB,StoicaStochasticCRB}.
\end{remark}

Substituting \eqref{eq:R_s_approx} into \eqref{eq:CRB_exact}, we obtain
\begin{equation}
\text{CRB}(\omega_{\mathrm{s},i}) = [\text{CRB}(\boldsymbol{\omega}_{\mathrm{s}})]_{ii}
\approx \sigma^2 / [2K \sigma_{\alpha,\mathrm{s}}^2 P h(\omega_{\mathrm{s},i})], \label{eq:CRB_Rs_approx}
\end{equation}
where $h(\omega_{\mathrm{s},i})$ represents the spatial sensitivity of the array, given by
\begin{equation}
h(\omega_{\mathrm{s},i}) = \mathbf{d}_{\mathrm{s}}^H(\omega_{\mathrm{s},i}) [\mathbf{I} - \mathbf{A}_{\mathrm{s}}(\mathbf{A}_{\mathrm{s}}^H \mathbf{A}_{\mathrm{s}})^{-1} \mathbf{A}_{\mathrm{s}}^H] \mathbf{d}_{\mathrm{s}}(\omega_{\mathrm{s},i}).
\label{eq:h_exact}
\end{equation}
Note that  $h(\omega_{\mathrm{s},i})$ can be interpreted as the squared norm of the projection of the steering vector derivative $\mathbf{d}_{\mathrm{s}}(\omega_{\mathrm{s},i})$ onto the null space of the signal subspace. 
\begin{comment}
Thus, for sources that are not closely spaced, the steering vectors are asymptotically orthogonal, and $h(\omega_{\mathrm{s},i})$ is dominated by the squared norm of the derivative vector:
\begin{equation}
h(\omega_{\mathrm{s},i}) \approx \|\mathbf{d}_{\mathrm{s}}(\omega_{\mathrm{s},i})\|^2 = \sum_{k=0}^{N_{\mathrm{s}}-1} n_{\mathrm{s},k}^2, \label{eq:h_approx}
\end{equation}
where $\{n_{\mathrm{s},0}, n_{\mathrm{s},1}, \dots, n_{\mathrm{s},N_{\mathrm{s}}-1}\}$ are the positions of antennas in units of half-wavelength. This shows how the angle CRB is related to the antenna array architecture.
\end{comment}
For well-separated sources, we assume that
$\mathbf{a}^{H}(\omega_{\mathrm{s},k})\mathbf{a}(\omega_{\mathrm{s},i})$ and $\mathbf{a}^{H}(\omega_{\mathrm{s},k})\mathbf{d}_{\mathrm{s}}(\omega_{\mathrm{s},i})$ are negligible for all $k\neq i$.
It then follows that
\begin{IEEEeqnarray}{rCl}
    h(\omega_{\mathrm{s},i})
    &\approx&
    \mathbf{d}_{\mathrm{s}}^H(\omega_{\mathrm{s},i}) [\mathbf{I} - N_{\mathrm{s}}^{-1} \mathbf{a}_{\mathrm{s}}(\omega_{\mathrm{s},i}) \mathbf{a}_{\mathrm{s}}^H(\omega_{\mathrm{s},i})] \mathbf{d}_{\mathrm{s}}(\omega_{\mathrm{s},i}) \IEEEnonumber \\
    &=&
    % \|\mathbf{d}(\omega_{\mathrm{s},i})\|^{2} - \frac{\left| \mathbf{a}^{H}(\omega_{\mathrm{s},i}) \mathbf{d}(\omega_{\mathrm{s},i}) \right|^{2}}{N_{\mathrm{s}}}
    % =
    \sum_{k=0}^{N_{\mathrm{s}}-1}
    \left(
    n_{\mathrm{s},k}
    -\bar n_{\mathrm{s}}
    \right)^{2}, \label{eq:h_approx}
\end{IEEEeqnarray}
where $\{n_{\mathrm{s},0}, n_{\mathrm{s},1}, \dots, n_{\mathrm{s},N_{\mathrm{s}}-1}\}$ are the positions of antennas in units of half-wavelength, and
$\bar n_{\mathrm{s}}
=
N_{\mathrm{s}}^{-1}
\sum_{k=0}^{N_{\mathrm{s}}-1}
n_{\mathrm{s},k}$.
This shows how the angle CRB is related to the array architecture, i.e., the second-order central moment of antenna positions.

\begin{remark}
    The results in \eqref{eq:h_approx} suggest that the angle CRB is asymptotically irrelevant to the array architecture of the transmitter, which is different from the capacity analysis in Sec. \ref{sec:capacity} where SAs at both the transmitter and communication receiver have their impacts. This is because the DOA estimation performance under equal-power transmission without complicated precoding is primarily determined by the receive antenna array of the sensing receiver, 
    % (provided that the reflected signal power from the targets is sufficiently strong)
    whereas the transmitter mainly serves to deliver enough power to excite target reflections.
\end{remark}

\subsection{Angle CRB Comparison Between SAs and ULAs} \label{ssc:CRB_SA_ULA}

With the angle CRB analysis result in the previous subsection, we compare the angle CRBs between ULAs and SAs.

\subsubsection{Angle CRB of ULAs}
We consider a ULA with $N_{\mathrm{s}}$ sensing receive antennas, where the positions are $n_{\mathrm{s},k} \in \{0, 1, \dots, N_{\mathrm{s}}-1\}$, leading to the aperture of $O(N_{\mathrm{s}})$. In addition, the sensitivity term scales as:
\begin{equation}
h_{\mathrm{ULA}}(\omega_{\mathrm{s},i}) \approx \sum_{k=0}^{N_{\mathrm{s}}-1} \left(k -N_{\mathrm{s}}'\right)^2 = \frac{N_{\mathrm{s}}^3}{12} + O(N_{\mathrm{s}}^2) = O(N_{\mathrm{s}}^3). \nonumber
\end{equation}
where $N_{\mathrm{s}}' = (N_{\mathrm{s}} - 1)/2$.
Hence, the order of the CRB for a ULA is
\begin{equation}
\text{CRB}_{\mathrm{ULA}}(\omega_{\mathrm{s},i}) = O\left( (K N_{\mathrm{s}}^3)^{-1} \right).
\label{eq:CRB_ULA}
\end{equation}

\subsubsection{Angle CRB of SAs}
When considering a SA with $N_{\mathrm{s}}$ antennas, we typically consider a quadratic aperture $O(N_{\mathrm{s}}^2)$, where antenna positions $n_{\mathrm{s},k}$ scale such that the maximum position $n_{\mathrm{s},N_{\mathrm{s}}-1} = O(N_{\mathrm{s}}^2)$, as in nested, coprime, and minimum redundancy arrays \cite{PiyaPPVnested,PPVCoprimePaper,Moffet}, so as to create a virtual ULA with length of $O(N_{\mathrm{s}}^2)$ using the difference coarray method \cite{PiyaPPVnested}. In this case, there are $O(N_{\mathrm{s}})$ terms in the summation of $h_{\mathrm{SA}}(\omega_{\mathrm{s},i})$ that are of order $(N_{\mathrm{s}}^2)^2 = N_{\mathrm{s}}^4$. 
% \begin{equation}
% h_{\mathrm{SA}}(\omega_{\mathrm{s},i}) \approx \sum_{k=0}^{N_{\mathrm{s}}-1} n_{\mathrm{s},k}^2
% \end{equation}
%that are of order $(N_{\mathrm{s}}^2)^2 = N_{\mathrm{s}}^4$. 
Thus, the sensitivity term scales as:
\begin{equation}
h_{\mathrm{SA}}(\omega_{\mathrm{s},i}) = O(N_{\mathrm{s}} \cdot (N_{\mathrm{s}}^2)^2) = O(N_{\mathrm{s}}^5).
\end{equation}
Hence, the order of the CRB for a typical SA is
\begin{equation}
\text{CRB}_{\mathrm{SA}}(\omega_{\mathrm{s},i}) = O\left( (K N_{\mathrm{s}}^5)^{-1} \right).
\label{eq:CRB_SA}
\end{equation}

The above analysis shows that since a SA can expand the aperture of the antenna array from $O(N_{\mathrm{s}})$ to $O(N_{\mathrm{s}}^2)$, while maintaining $N_{\mathrm{s}}$ physical antennas, the angle CRB can be improved from $O\left(1 / N_{\mathrm{s}}^3 \right)$ to $O\left(1 / N_{\mathrm{s}}^5 \right)$.
\begin{example}[Angle CRB of nested arrays]
    Taking an $N_{\mathrm{s}}$-antenna nested array \cite{PiyaPPVnested} given in Definition \ref{def:SA} as an example, we can show that the dominant term of $h_{\mathrm{SA}}(\omega_{\mathrm{s},i})$ is $5 N_{\mathrm{s}}^5/768$.
\end{example}

\section{MIMO Systems With Modern Waveforms} \label{sec:waveforms}

In this section, we extend the snapshot-based spatial MIMO model in Sec. \ref{sec:model} to models with two common modern waveforms OTFS and OFDM. We show that the results derived earlier still hold when these waveforms are incorporated. We consider MIMO-OTFS and MIMO-OFDM ISAC systems similar to the system in Sec. \ref{sec:model}.
Each scattering path $i$ is also associated with range and Doppler velocity $(R_{\mathrm{t},i}, v_{\mathrm{t},i})$ and $(R_{\mathrm{r},i}, v_{\mathrm{r},i})$ with respect to the transmitter and communication receiver, respectively.

\subsection{Communication Signal Models for MIMO-OTFS and MIMO-OFDM ISAC Systems}
For the standard MIMO-OTFS system \cite{Srivastava,Raviteja2019} and MIMO-OFDM system \cite{XiaoSwindlehurst}, the vectorized received communication signal within a CPI can be expressed as
\begin{equation}
    \mathbf{y}_{z,\mathrm{c}} = \mathbf{H}_{z,\mathrm{c}} \mathbf{x}_{z} + \mathbf{w}_{z,\mathrm{c}} \in \mathbb{C}^{N_{\mathrm{r}} LM \times 1},
    \label{eq:MIMO_WF_c}
\end{equation}
where $z \in \{\mathrm{DD}, \mathrm{TF}\}$ represents delay-Doppler-domain OTFS and time-frequency-domain OFDM, respectively, $M$ subsymbols and $L$ subcarriers are considered for the signal, and
\begin{equation}
    \mathbf{H}_{z,\mathrm{c}} = \sum_{i=1}^{D_{\mathrm{c}}} \alpha_{\mathrm{c},i} (\mathbf{S}_{\mathrm{c},i} \otimes \mathbf{U}_{z,\mathrm{c},i}) \label{eq:Hzc}
\end{equation}
is the effective MIMO-OTFS or MIMO-OFDM communication channel matrix, in which $\alpha_{\mathrm{c},i}$ is the channel coefficient of scatterer $i$; $\mathbf{x}_{z} = \text{vec}(\mathbf{X}_{z}^T) \in \mathbb{C}^{N_{\mathrm{t}} LM \times 1}$ is the vectorized transmitted symbols from the $N_{\mathrm{t}} \times LM$ delay-Doppler or time-frequency grid $\mathbf{X}_{z}$; $\mathbf{w}_{z,\mathrm{c}}$ is the additive complex white Gaussian noise (AWGN) at delay-Doppler or time-frequency domain of the communication receiver; $\mathbf{U}_{z,\mathrm{c},i} \in \mathbb{C}^{LM \times LM}$ is the unitary OTFS or OFDM operator (delay-Doppler shift) for communication defined in (\ref{eq:UDDci}) or (\ref{eq:UTFci}); and $\mathbf{S}_{\mathrm{c},i} = \mathbf{a}_{\mathrm{r}}(\omega_{\mathrm{r},i})\mathbf{a}_{\mathrm{t}}^H(\omega_{\mathrm{t},i}) \in \mathbb{C}^{N_{\mathrm{r}} \times N_{\mathrm{t}}}$ is the spatial signature of scatterer $i$, as defined in \eqref{eq:H_c_k}.

For the MIMO-OTFS system,
\begin{equation}
\mathbf{U}_{\mathrm{DD,c},i} = (\mathbf{F}_{M} \otimes \mathbf{I}_{L}) \left( \bm{\Pi}^{l_{\mathrm{c},i}} \bm{\Delta}^{k_{\mathrm{c},i}} \right) (\mathbf{F}_{M}^{H} \otimes \mathbf{I}_{L}) \label{eq:UDDci}
\end{equation}
is a unitary matrix, where $\mathbf{F}_{M} \in \mathbb{C}^{M \times M}$ is the (normalized) $M$-point Discrete Fourier Transform (DFT) matrix; $\bm{\Pi} \in \mathbb{R}^{LM \times LM}$ is the forward cyclic shift permutation matrix, defined such that $[\bm{\Pi}]_{p,q} = 1$ if $p = (q+1) \pmod{LM}$ and $0$ otherwise; and $\bm{\Delta} \in \mathbb{C}^{LM \times LM}$ is a diagonal matrix with diagonal entries given by $[\bm{\Delta}]_{n,n} = e^{j2\pi \frac{n}{LM}}$ for $n=0, 1, \dots, LM-1$.
The delay shift $l_{\mathrm{c},i} = \frac{\tau_{\mathrm{c},i}}{\Delta\tau}$ is assumed to be an integer, where $\Delta\tau = 1 / (L \Delta f)$ is the delay grid spacing, $\Delta f$ is the subcarrier spacing, and $\tau_{\mathrm{c},i} = (R_{\mathrm{t},i} + R_{\mathrm{r},i}) / c$ is the delay of scattering path $i$. We assume no fractional delay shift for simplicity, as the delay resolution $\Delta\tau = 1 / (L \Delta f)$ is commonly sufficient to approximate the path delays in typical wide-band systems \cite{TseWirelessComm,Raviteja2018}.
On the other hand, the Doppler shift $k_{\mathrm{c},i} = \frac{\nu_{\mathrm{c},i}}{\Delta\nu}$ could be either an integer or non-integer, where $\Delta\nu = \Delta f / M$ is the Doppler grid spacing, $\nu_{\mathrm{c},i} = (v_{\mathrm{t},i} + v_{\mathrm{r},i})f_c/c$ is the Doppler frequency shift for the path through scatterer $i$, $f_c = c/\lambda_c$ is the carrier frequency, and $c$ is the speed of light. When $k_{\mathrm{c},i}$ is a non-integer, corresponding to the fractional Doppler case, $\bm{\Delta}^{k_{\mathrm{c},i}}$ should be understood as a diagonal matrix with diagonal entries given by $[\bm{\Delta}^{k_{\mathrm{c},i}}]_{n,n} = e^{j2\pi \frac{nk_{\mathrm{c},i}}{LM}}$ for $n=0, 1, \dots, LM-1$.

For the MIMO-OFDM system, we assume that $T = T_{\mathrm{d}} + T_{\mathrm{cp}}$ is the total symbol duration, $T_{\mathrm{d}} = 1 / \Delta f$ is the OFDM symbol duration, and the cyclic-prefix (CP) duration $T_{\mathrm{cp}}$ is long enough so that there is no inter-symbol interference. Then, after CP removal and DFT, the received signal follows \eqref{eq:MIMO_WF_c} and \eqref{eq:Hzc}, where $\mathbf{U}_{\mathrm{TF,c},i}$ is a diagonal matrix whose diagonal terms are
\begin{equation}
    \begin{aligned}
        &[\mathbf{U}_{\mathrm{TF,c},i}]_{l+mL,l+mL} \\
&= e^{-j2\pi l\Delta f(R_{\mathrm{t},i} + R_{\mathrm{r},i})/c} e^{j2\pi mT(v_{\mathrm{t},i} + v_{\mathrm{r},i})f_c/c} \\
    &= e^{-j2\pi l_{\mathrm{c},i} l / L} e^{j2\pi k_{\mathrm{c},i} T \Delta f m / M}
    \end{aligned} \label{eq:UTFci}
\end{equation}
for $0\leq l \leq L-1$ and $0\leq m \leq M-1$.

\subsection{Sensing Signal Models for MIMO-OTFS and MIMO-OFDM ISAC Systems}

We suppose that there exist $D_{\mathrm{s}}$ targets for the sensing channel, and denote $(R_{\mathrm{s},i}, v_{\mathrm{s},i}, \theta_{\mathrm{s},i})$ as the range, Doppler velocity, and angle of target $i$ for monostatic sensing in Fig. \ref{fig:MIMO_channel}. Then, the vectorized received sensing signal for the MIMO-OTFS system \cite{Srivastava,Raviteja2019} and MIMO-OFDM system \cite{XiaoSwindlehurst} under monostatic sensing can be expressed as
\begin{equation}
    \mathbf{y}_{z,\mathrm{s}} = \mathbf{H}_{z,\mathrm{s}} \mathbf{x}_{z} + \mathbf{w}_{z,\mathrm{s}} \in \mathbb{C}^{N_{\mathrm{s}} LM \times 1}, 
    \label{eq:MIMO_WF_s}
\end{equation}
where $z \in \{\mathrm{DD}, \mathrm{TF}\}$ represents OTFS and OFDM, respectively, $\mathbf{H}_{z,\mathrm{s}} = \sum_{i=1}^{D_{\mathrm{s}}} \alpha_{\mathrm{s},i} (\mathbf{S}_{\mathrm{s},i} \otimes \mathbf{U}_{z,\mathrm{s},i})$ is the MIMO-OTFS or MIMO-OFDM sensing channel matrix, $\mathbf{w}_{z,\mathrm{s}}$ is the AWGN at delay-Doppler or time-frequency domain of the sensing receiver, $\alpha_{\mathrm{s},i}$ is the channel coefficient for the sensing path of target $i$, $\mathbf{S}_{\mathrm{s},i} = \mathbf{a}_{\mathrm{s}}(\omega_{\mathrm{s},i})\mathbf{a}_{\mathrm{t}}^H(\omega_{\mathrm{s},i}) \in \mathbb{C}^{N_{\mathrm{s}} \times N_{\mathrm{t}}}$ is the spatial signature of target $i$, with $\omega_{\mathrm{s},i} = 2\pi d\sin\theta_{\mathrm{s},i}/\lambda_c$ being the angle of target $i$ and $\mathbf{a}_{\mathrm{s}}(\omega_{\mathrm{s},i})$ follows again the expression in \eqref{eq:a_q_w} in consideration of the received sensing antenna array architecture.
For MIMO-OTFS,
\begin{equation}
\mathbf{U}_{\mathrm{DD,s},i} = (\mathbf{F}_{M} \otimes \mathbf{I}_{L}) \left( \bm{\Pi}^{l_{\mathrm{s},i}} \bm{\Delta}^{k_{\mathrm{s},i}} \right) (\mathbf{F}_{M}^{H} \otimes \mathbf{I}_{L}), \label{eq:U_DD_s_i}
\end{equation}
where $l_{\mathrm{s},i} = \frac{\tau_{\mathrm{s},i}}{\Delta\tau} = 2L \Delta f R_{\mathrm{s},i} / c$ is the delay shift of target $i$ and $k_{\mathrm{s},i} = \frac{\nu_{\mathrm{s},i}}{\Delta\nu} = 2\frac{M}{\Delta f} v_{\mathrm{s},i}f_c/c$ is the Doppler shift of target $i$. Following the same approach as for the communication channel, we again assume an integer delay shift.
For MIMO-OFDM, $\mathbf{U}_{\mathrm{TF,s},i}$ is a diagonal matrix whose diagonal terms are
\begin{equation}
    \begin{aligned}  \relax [\mathbf{U}_{\mathrm{TF,s},i}]_{l+mL,l+mL} &= e^{-j4\pi l\Delta fR_{\mathrm{s},i}/c} e^{j4\pi mTv_{\mathrm{s},i}f_c/c} \\
    &= e^{-j2\pi l_{\mathrm{s},i} l / L} e^{j2\pi k_{\mathrm{s},i} T \Delta f m / M}.
    \end{aligned} \label{eq:U_TF_s_i}
\end{equation}
By comparing \eqref{eq:MIMO_WF_s} with \eqref{eq:MIMO_WF_c}, it is clear that they share the same expression, with the only difference being that monostatic sensing assumes identical DODs and DOAs.

\subsection{Channel Capacity of SAs and ULAs in MIMO-OTFS and MIMO-OFDM Systems}

In this subsection, we apply the proposed majorization method in Sec. \ref{sec:capacity} for the MIMO-OTFS and MIMO-OFDM systems.
We consider the channel $\mathbf{H}_{z,\mathrm{c}}$, $z \in \{\mathrm{DD}, \mathrm{TF}\}$, in (\ref{eq:Hzc}).
The Gram matrix $\mathbf{G}_{z,\mathrm{c}} = \mathbf{H}_{z,\mathrm{c}} \mathbf{H}_{z,\mathrm{c}}^H$ is
\begin{align}
\mathbf{G}_{z,\mathrm{c}} &= \sum_{i=1}^{D_{\mathrm{c}}} \alpha_{\mathrm{c},i} (\mathbf{S}_{\mathrm{c},i} \otimes \mathbf{U}_{z,\mathrm{c},i}) \sum_{j=1}^{D_{\mathrm{c}}} \alpha_{\mathrm{c},j}^* (\mathbf{S}_{\mathrm{c},j} \otimes \mathbf{U}_{z,\mathrm{c},j}). 
\end{align}
Taking the expectation over the path gains and noting that $\E[\alpha_{\mathrm{c},i} \alpha_{\mathrm{c},j}^*] = \sigma_{\alpha,\mathrm{c}}^2 \delta_{ij}$, we have
\begin{IEEEeqnarray}{rCl}
\E[\mathbf{G}_{z,\mathrm{c}}] &=& \E \left[ \sum_{i=1}^{D_{\mathrm{c}}} \sigma_{\alpha,\mathrm{c}}^2 (\mathbf{S}_{\mathrm{c},i} \otimes \mathbf{U}_{z,\mathrm{c},i}) (\mathbf{S}_{\mathrm{c},i} \otimes \mathbf{U}_{z,\mathrm{c},i})^H \right]. \IEEEeqnarraynumspace
\end{IEEEeqnarray}
Since $\mathbf{U}_{z,\mathrm{c},i}$ is a unitary matrix for any $l_{\mathrm{c},i}$, $k_{\mathrm{c},i}$, and $i$, we have
\begin{IEEEeqnarray}{rCl}
\E[\mathbf{G}_{z,\mathrm{c}}] &=& \E \left[ \sum_{i=1}^{D_{\mathrm{c}}} \sigma_{\alpha,\mathrm{c}}^2 \mathbf{S}_{\mathrm{c},i} \mathbf{S}_{\mathrm{c},i}^H \right] \otimes \mathbf{I}_{LM} = \E[\mathbf{G}] \otimes \mathbf{I}_{LM} \IEEEeqnarraynumspace
\end{IEEEeqnarray}
where $\E[\mathbf{G}]$ is as given in (\ref{eq:E_G}).
Using \eqref{eq:E_G_R}, we obtain
\begin{IEEEeqnarray}{rCl}
\E[\mathbf{G}_{z,\mathrm{c}}] &=& D_{\mathrm{c}} \sigma_{\alpha,\mathrm{c}}^2 N_{\mathrm{t}} \mathbf{R}_{\mathrm{r}} \otimes \mathbf{I}_{LM},\label{eq:E_G_zc}
\end{IEEEeqnarray}
where the entries $R_{\mathrm{r},k,m}$ of $\mathbf{R}_{\mathrm{r}}$ are given by \eqref{eq:R_rkm}.
%\begin{equation}
%R_{\mathrm{r},k,m} = \E[e^{j\omega (n_{\mathrm{r},k} - n_{\mathrm{r},m})}].
%\end{equation}
%in the general case, and are given by (\ref{eq:Rkm}) if the probability density function $f_{\Omega}(\omega)$ of $\omega$ follows the special case (\ref{eq:PDF_omega}).
Note that we have
\begin{IEEEeqnarray}{rCl}
\boldsymbol{\lambda}(\E[\mathbf{G}_{z,\mathrm{c}}]) &=& \boldsymbol{\lambda}(\E[\mathbf{G}]) \otimes \mathbf{1}_{LM} ,\label{eq:eig_E_G_zc}
\end{IEEEeqnarray}
where $\mathbf{1}_{LM}$ is a column vector of ones of length $LM$.
That is, each spatial eigenvalue is simply repeated $LM$ times in the overall MIMO-OTFS and MIMO-OFDM systems.
This result shows that the delay-Doppler-domain or time-frequency-domain processing in OTFS or OFDM (represented by the unitary operators $\mathbf{U}_{z,\mathrm{c},i}$) does not change the ``spread" or distribution of the spatial eigenvalues on average; it merely scales the total number of available modes by the delay-Doppler or time-frequency resource factor $LM$.

\begin{remark}\label{rmk:same_capacity}
    Our results in \eqref{eq:E_G_zc} and \eqref{eq:eig_E_G_zc} reveal that the {\it ergodic capacities} of both MIMO-OTFS and MIMO-OFDM systems are governed by the array structure within a single resource block, i.e., $\mathbf{G}$, as well as the number $LM$ of resource blocks. Therefore, the comparisons between SAs and ULAs, along with the conclusions established in Sections~\ref{ssc:ULA_mjrz_SA} and \ref{ssc:edof}, can be directly extended to both MIMO-OTFS and MIMO-OFDM systems. However, these results should be used solely for comparing different array structures and should not be interpreted as a comparison between MIMO-OTFS and MIMO-OFDM systems. This is because ergodic capacity characterizes only the fundamental information-theoretic limit and does not account for receiver complexity or finite blocklength performance, both of which must be considered when comparing different waveforms.
    Moreover, the OFDM model \eqref{eq:UTFci} is based on assumptions not required for OTFS, including no Doppler shifts or inter-carrier interference within a subsymbol, and one CP for each subsymbol to mitigate inter-subsymbol interference.
\end{remark}

\subsection{Angle CRB of SAs and ULAs in MIMO-OTFS and MIMO-OFDM Systems}

Next, we show how to extend the angle CRB analysis in Sec. \ref{sec:CRB} to the MIMO-OTFS and MIMO-OFDM systems.
We consider the vectorized received sensing signal $\mathbf{y}_{z,\mathrm{s}}$, $z \in \{\mathrm{DD}, \mathrm{TF}\}$, for MIMO-OTFS and MIMO-OFDM systems in \eqref{eq:MIMO_WF_s}.
Using the identity that $(\mathbf{B} \otimes \mathbf{A})\text{vec}(\mathbf{Z}) = \text{vec}(\mathbf{A}\mathbf{Z}\mathbf{B}^T)$ and noting that $\mathbf{x}_{z} = \text{vec}(\mathbf{X}_{z}^T)$, we obtain $(\mathbf{S}_{\mathrm{s},i} \otimes \mathbf{U}_{z,\mathrm{s},i}) \text{vec}(\mathbf{X}_{z}^T) = \text{vec}(\mathbf{U}_{z,\mathrm{s},i} \mathbf{X}_{z}^T \mathbf{S}_{\mathrm{s},i}^T)$, leading to
\begin{IEEEeqnarray}{rCl}
    \mathbf{y}_{z,\mathrm{s}} = \text{vec} \left( \sum_{i=1}^{D_{\mathrm{s}}} \alpha_{\mathrm{s},i} \mathbf{U}_{z,\mathrm{s},i} \mathbf{X}_{z}^T \mathbf{a}_{\mathrm{t}}^*(\omega_{\mathrm{s},i})\mathbf{a}_{\mathrm{s}}^T(\omega_{\mathrm{s},i}) \right) + \mathbf{w}_{z,\mathrm{s}}. \IEEEnonumber
\end{IEEEeqnarray}
We then let $\mathbf{Y}_{z,\mathrm{s}}, \mathbf{W}_{z,\mathrm{s}} \in \mathbb{C}^{N_{\mathrm{s}} \times LM}$ be the data and noise matrix such that $\mathbf{y}_{z,\mathrm{s}} = \text{vec}(\mathbf{Y}_{z,\mathrm{s}}^T)$ and $\mathbf{w}_{z,\mathrm{s}} = \text{vec}(\mathbf{W}_{z,\mathrm{s}}^T)$. By removing the $\text{vec}(\cdot)$ operator, we obtain
\begin{IEEEeqnarray}{rCl}
    \mathbf{Y}_{z,\mathrm{s}}^T = \sum_{i=1}^{D_{\mathrm{s}}} \underbrace{\left( \alpha_{\mathrm{s},i} \mathbf{U}_{z,\mathrm{s},i} \mathbf{X}_{z}^T \mathbf{a}_{\mathrm{t}}^*(\omega_{\mathrm{s},i}) \right)}_{=\mathbf{s}_{z,i} \in \mathbb{C}^{LM \times 1}} \mathbf{a}_{\mathrm{s}}^T(\omega_{\mathrm{s},i}) + \mathbf{W}_{z,\mathrm{s}}^T, \IEEEnonumber
\end{IEEEeqnarray}
in which we denote the signal vector for target $i$ as $\mathbf{s}_{z,i}=\alpha_{\mathrm{s},i} \mathbf{U}_{z,\mathrm{s},i} \mathbf{X}_{z}^T \mathbf{a}_{\mathrm{t}}^*(\omega_{\mathrm{s},i})$, and its $t$-th element $s_{z,i}(t)$ represents the complex amplitude of the $i$-th target at the $t$-th Delay-Doppler bin, with $t=1, \dots, LM$. Hence, we can regard the $t$-th column of $\mathbf{Y}_{z,\mathrm{s}}$, denoted by $\mathbf{y}_{z,\mathrm{s}}(t)$, as the received signal at snapshot $t$, yielding
\begin{equation}
\begin{aligned}
    \mathbf{y}_{z,\mathrm{s}}(t) &=\sum_{i=1}^{D_{\mathrm{s}}} \mathbf{a}_{\mathrm{s}}(\omega_{\mathrm{s},i}) s_{z,i}(t) + \mathbf{w}_{z,\mathrm{s}}(t)\\
    &=\mathbf{A}_{\mathrm{s}}(\boldsymbol{\omega}_{\mathrm{s}}) \mathbf{s}_{z}(t) + \mathbf{w}_{z,\mathrm{s}}(t)
    \label{eq:sensing_model_WF_DOA}    
\end{aligned}
\end{equation}
where $\mathbf{s}_{z}(t) = [s_{z,1}(t)\; \cdots\; s_{z,D_{\mathrm{s}}}(t)]^T$ is the signal vector at snapshot $t$, and $\mathbf{w}_{z,\mathrm{s}}(t)$ is the $t$-th column of $\mathbf{W}_{z,\mathrm{s}}$.

As in Sec. \ref{sec:CRB}, by using the \textit{conditional model} for CRB \cite{LiuCRB,StoicaStochasticCRB}, the derivation of angle CRB conditioning on the transmitted symbols $\mathbf{x}_{z}$ is related to finding the empirical source covariance matrix
\begin{equation}
    \hat{\mathbf{R}}_{z,\mathrm{s}} = \frac{1}{LM} \sum_{t=1}^{LM} \mathbf{s}_{z}(t)\mathbf{s}_{z}^H(t)\in \mathbb{C}^{D_{\mathrm{s}} \times D_{\mathrm{s}}}, \label{eq:R_zs}
\end{equation}
where the $(i, j)$-th entry of this matrix is
\begin{IEEEeqnarray}{rCl}
    [\hat{\mathbf{R}}_{z,\mathrm{s}}]_{i,j} &=& \frac{1}{LM} \sum_{t=1}^{LM} s_{z,i}(t)s_{z,j}^*(t)= \frac{1}{LM} \mathbf{s}_{z,j}^H \mathbf{s}_{z,i} \IEEEnonumber \\
    &=& \frac{\alpha_{\mathrm{s},i} \alpha_{\mathrm{s},j}^*}{LM} \mathbf{a}_{\mathrm{t}}^T(\omega_{\mathrm{s},j}) \mathbf{X}_{z}^* \mathbf{U}_{z,\mathrm{s},j}^H \mathbf{U}_{z,\mathrm{s},i} \mathbf{X}_{z}^T \mathbf{a}_{\mathrm{t}}^*(\omega_{\mathrm{s},i}) \IEEEnonumber \\
    &=& \frac{\alpha_{\mathrm{s},i} \alpha_{\mathrm{s},j}^*}{LM} \sum_{p=1}^{LM} \sum_{q=1}^{LM} [\mathbf{U}_{z,\mathrm{s},j}^H \mathbf{U}_{z,\mathrm{s},i}]_{p,q} b_{z,i,j,p,q} \label{eq:RDD_sij}
\end{IEEEeqnarray}
where $b_{z,i,j,p,q} = \mathbf{a}_{\mathrm{t}}^H(\omega_{\mathrm{s},i}) \mathbf{x}_{z,q} \mathbf{x}_{z,p}^H \mathbf{a}_{\mathrm{t}}(\omega_{\mathrm{s},j})$, and $\mathbf{x}_{z,t}$ is the $t$-th column of $\mathbf{X}_{z}$. 
When the transmitted symbols in $\mathbf{X}_{z}$ are considered as realizations of independent and identically distributed (i.i.d.) random variables with $E[\mathbf{X}_{z}] = \mathbf{0}$ and $E[\text{vec}(\mathbf{X}_{z}) \text{vec}(\mathbf{X}_{z})^H] = \frac{P}{N_{\mathrm{t}}} \mathbf{I}_{N_{\mathrm{t}} LM}$, we know that $\mathbf{x}_{z,t}$ is also i.i.d. with $E[\mathbf{x}_{z,p} \mathbf{x}_{z,q}^H] = \delta_{p,q} \frac{P}{N_{\mathrm{t}}} \mathbf{I}_{N_{\mathrm{t}}}$. Then, by following the law of large numbers with $LM \to \infty$, the sample mean converges to the expected value, leading to
\begin{equation}
    \begin{aligned}
     &[\hat{\mathbf{R}}_{z,\mathrm{s}}]_{i,j} \\
     &\stackrel{LM \to \infty}{=}%\xrightarrow{LM \to \infty}
     \frac{\alpha_{\mathrm{s},i} \alpha_{\mathrm{s},j}^*}{LM} \sum_{p=1}^{LM} [\mathbf{U}_{z,\mathrm{s},j}^H \mathbf{U}_{z,\mathrm{s},i}]_{p,p} \left( \frac{P}{N_{\mathrm{t}}} \mathbf{a}_{\mathrm{t}}^H(\omega_{\mathrm{s},i}) \mathbf{a}_{\mathrm{t}}(\omega_{\mathrm{s},j}) \right)\\
     &=\frac{\alpha_{\mathrm{s},i} \alpha_{\mathrm{s},j}^* P}{N_{\mathrm{t}} LM} \text{tr}(\mathbf{U}_{z,\mathrm{s},j}^H \mathbf{U}_{z,\mathrm{s},i}) \mathbf{a}_{\mathrm{t}}^H(\omega_{\mathrm{s},i}) \mathbf{a}_{\mathrm{t}}(\omega_{\mathrm{s},j}).
    \end{aligned} \nonumber
\end{equation}

For MIMO-OTFS, we consider $\mathbf{U}_{\mathrm{DD,s},i}$ in \eqref{eq:U_DD_s_i}. Since the trace is invariant under the unitary transformation $(\mathbf{F}_{M} \otimes \mathbf{I}_{L})$, by applying the cyclic property of the trace, we obtain
\begin{equation}
    \text{tr}(\mathbf{U}_{\mathrm{DD,s},j}^H \mathbf{U}_{\mathrm{DD,s},i}) = \text{tr}\left( \bm{\Delta}^{-k_{\mathrm{s},j}} \bm{\Pi}^{l_{\mathrm{s},i} - l_{\mathrm{s},j}} \bm{\Delta}^{k_{\mathrm{s},i}} \right).
\end{equation}
Recall that the delay shifts $l_{\mathrm{s},i} = \frac{\tau_{\mathrm{s},i}}{\Delta\tau} = 2L \Delta f R_{\mathrm{s},i} / c$ are assumed to be integers for all $i$, due to the use of wideband system (i.e., $L \to \infty$), we can obtain\footnote{Indeed, we can also show that the trace vanishes for any $i \neq j$ if we also assume integer Doppler shifts such that $k_{\mathrm{s},i} - k_{\mathrm{s},j}$ is a nonzero integer, but we do not require this.}
\begin{equation}
    \text{tr}(\mathbf{U}_{\mathrm{DD,s},j}^H \mathbf{U}_{\mathrm{DD,s},i}) = \sum_{n=0}^{LM-1} [\bm{\Pi}^{l_{\mathrm{s},i} - l_{\mathrm{s},j}}]_{n,n} e^{j2\pi \frac{n(k_{\mathrm{s},i} - k_{\mathrm{s},j})}{LM}} = 0 \nonumber
\end{equation}
$\forall i \neq j$. Consequently, $[\hat{\mathbf{R}}_{\mathrm{DD,s}}]_{i,j} = 0,\forall i \neq j$ as $L\to\infty$, showing that MIMO-OTFS provides an inherent mathematical decorrelation of targets on the angle domain. On the other hand, when $i = j$, we have $\text{tr}(\mathbf{U}_{\mathrm{DD,s},i}^H \mathbf{U}_{\mathrm{DD,s},i}) = LM$. As a result, we can conclude that for large $L$,
\begin{equation}
\hat{\mathbf{R}}_{\mathrm{DD,s}} \approx \text{diag}\left( |\alpha_{\mathrm{s},1}|^2 P, \dots, |\alpha_{\mathrm{s},D_{\mathrm{s}}}|^2 P \right),\label{eq:R_DDs_approx}
\end{equation}
where the approximation becomes exact almost surely as $L\to\infty$.

For MIMO-OFDM, we consider $\mathbf{U}_{\mathrm{TF,s},i}$ in \eqref{eq:U_TF_s_i}. Since $\mathbf{U}_{\mathrm{TF,s},i}$ is a diagonal matrix, we can obtain
% \begin{IEEEeqnarray}{rCl}
%     \text{tr}(\mathbf{U}_{\mathrm{TF,s},j}^H \mathbf{U}_{\mathrm{TF,s},i}) &=& \frac{1}{L} \sum_{l=0}^{L-1} e^{-j\frac{4\pi \Delta f \Delta R_{ij}}{c} l} \IEEEnonumber \\
%     && {}\cdot \frac{1}{M} \sum_{m=0}^{M-1} e^{j\frac{4\pi T f_c \Delta v_{ij}}{c} m}, \IEEEnonumber
% \end{IEEEeqnarray}
\begin{IEEEeqnarray}{rCl}
    \text{tr}(\mathbf{U}_{\mathrm{TF,s},j}^H \mathbf{U}_{\mathrm{TF,s},i}) &=& \sum_{l=0}^{L-1} e^{-j\frac{4\pi \Delta f \Delta R_{ij}}{c} l} \sum_{m=0}^{M-1} e^{j\frac{4\pi T f_c \Delta v_{ij}}{c} m}, \IEEEnonumber
\end{IEEEeqnarray}
where $\Delta R_{ij} = R_{\mathrm{s},i} - R_{\mathrm{s},j}$ and $\Delta v_{ij} = v_{\mathrm{s},i} - v_{\mathrm{s},j}$.
Using the geometric series sum formula for the $l$-sum, we get
\begin{equation}
\left\vert \frac{1}{L} \sum_{l=0}^{L-1} e^{-j\frac{4\pi \Delta f \Delta R_{ij}}{c} l} \right\vert=\left| \frac{\sin(2L \pi \Delta f \Delta R_{ij} / c)}{L \sin(2\pi \Delta f \Delta R_{ij} / c)} \right|.
\label{eq:l_sum}
\end{equation}
Thus, as $L \to \infty$, \eqref{eq:l_sum} vanishes for any $\Delta R_{ij} \neq 0$. More precisely, if we have $2L \Delta f \Delta R_{ij} / c \gg 1$ for any $i \neq j$, which is indeed true in typical wide-band systems whose delay resolution $\Delta\tau = 1 / (L \Delta f)$ is sufficient to approximate the path delays to the nearest sampling points \cite{TseWirelessComm,Raviteja2018}, then \eqref{eq:l_sum} vanishes.
Equivalently, if we use the assumption that the delay shifts $l_{\mathrm{s},i} = \frac{\tau_{\mathrm{s},i}}{\Delta\tau} = 2L \Delta f R_{\mathrm{s},i} / c$ are integers for all $i$ (and $l_{\mathrm{s},i} \neq l_{\mathrm{s},j}$ for $i \neq j$), we can also obtain that \eqref{eq:l_sum} is zero.\footnote{Similarly, the $m$-sum vanishes for any $i \neq j$ if we also assume $M \to \infty$ or integer Doppler shifts, but we do not require this.} As a consequence, the off-diagonal elements of $\hat{\mathbf{R}}_{\mathrm{TF,s}}$ should converge to zero if $L\to\infty$, yielding a diagonal source covariance matrix given by
\begin{equation}
\hat{\mathbf{R}}_{\mathrm{TF,s}} \approx \text{diag}\left( |\alpha_{\mathrm{s},1}|^2 P, \dots, |\alpha_{\mathrm{s},D_{\mathrm{s}}}|^2 P \right), \label{eq:R_TFs_approx}
\end{equation}
where the approximation becomes exact almost surely as $L\to\infty$.
Simulations in Sec. \ref{sec:sim} show that (\ref{eq:R_DDs_approx}) and (\ref{eq:R_TFs_approx}) are highly accurate even for a moderate number of subcarriers, such as $L=16$.

As in Sec. \ref{sec:CRB}, we can derive the conditional CRB for the DOA parameters $\boldsymbol{\omega}_{\mathrm{s}}$ for the MIMO-OTFS and MIMO-OFDM systems as
\begin{equation}
    \text{CRB}_z(\boldsymbol{\omega}_{\mathrm{s}}) = \frac{\sigma^2}{2LM} \left[ \text{Re} \left\{ (\mathbf{D}_{\mathrm{s}}^H \mathbf{P}_{\mathbf{A}_{\mathrm{s}}}^{\perp} \mathbf{D}_{\mathrm{s}}) \odot \hat{\mathbf{R}}_{z,\mathrm{s}}^T \right\} \right]^{-1}, 
    \label{eq:CRB__WF_exact}
\end{equation}
$z \in \{\mathrm{DD}, \mathrm{TF}\}$.
Using (\ref{eq:R_DDs_approx}) and (\ref{eq:R_TFs_approx}), we can obtain
\begin{equation}
\text{CRB}_z(\omega_{\mathrm{s},i}) = [\text{CRB}_z(\boldsymbol{\omega}_{\mathrm{s}})]_{ii}
\approx \frac{\sigma^2}{2LM P_i h(\omega_{\mathrm{s},i})}, \label{eq:CRB_Rs_WF_approx}
\end{equation}
where $P_i = |\alpha_{\mathrm{s},i}|^2 P$ is the effective signal power, so the two systems have the same CRB.
Comparing \eqref{eq:CRB_Rs_WF_approx} to \eqref{eq:CRB_Rs_approx}, we see that like the MIMO system, the CRB for MIMO-OTFS and MIMO-OFDM systems also depends on $h(\omega_{\mathrm{s},i})$ in \eqref{eq:h_exact}, so the angle CRB comparison between SAs and ULAs in Sec. \ref{ssc:CRB_SA_ULA} also applies to MIMO-OTFS and MIMO-OFDM systems.

\section{Simulations} \label{sec:sim}
In this section, we present numerical simulation results to validate our theoretical analysis of the fundamental ISAC limits for the MIMO system defined in Sec. \ref{sec:model}, and MIMO-OTFS and MIMO-OFDM systems defined in Sec. \ref{sec:waveforms}. We specifically focus on the impact of the geometry of the antenna array on the ergodic channel capacity and the CRB for angle estimation.
% For all simulations, we consider high-dynamics mmWave ISAC channels, as given in Sec. \ref{sec:model}.
We use $K=32$ snapshots for the model in Sec. \ref{sec:model}.
The carrier frequency is $f_c = 28$ GHz, the subcarrier spacing is $\Delta f = 120$ kHz, the OFDM symbol duration is $T_{\mathrm{d}} = 8.33 \mu s$, the cyclic-prefix duration is $T_{\mathrm{cp}} = 0.59 \mu s$, and the antenna unit spacing is $d = 0.5 c / f_c$.
In the geometric MIMO channel model, the path gains are i.i.d. zero-mean unit-variance circularly symmetric complex Gaussian, and all DODs and DOAs are i.i.d. uniformly distributed over the interval $(-\omega_{\text{max}}, \omega_{\text{max}}]$.
Unless otherwise specified, we use the following simulation parameters. The total number of subcarriers and OTFS or OFDM symbols are set to $L=16$ and $M=2$, respectively.
The numbers of antennas are $N_{\mathrm{t}} = N_{\mathrm{r}} = N_{\mathrm{s}} = 16$ for the ISAC transmitter, communication receiver, and sensing receiver. The number of communication paths is $D_{\mathrm{c}} = 8$. The number of sensing targets is $D_{\mathrm{s}} = 3$.
In Monte Carlo simulations, the delay shifts $l_{\mathrm{c},i}$ and $l_{\mathrm{s},i}$ for all communication and sensing paths are realizations of i.i.d. uniformly distributed integers, and the Doppler shifts $k_{\mathrm{c},i}$ and $k_{\mathrm{s},i}$ are realizations of i.i.d. uniformly distributed real numbers. 
We assume that the three arrays used at the ISAC transmitter, communication receiver, and sensing receiver are the same.
To compare with an $N$-antenna ULA, we consider a nested array (NA) \cite{PiyaPPVnested} and a coprime array (CA) \cite{PPVCoprimePaper} given in Definition \ref{def:SA}, where $M_{\mathrm{co}} = N/2 - 1$ for the CA.
The SNR is defined as $\rho = P/(N_t\sigma^2)$, where equal power allocation is assumed across all transmit antennas and subchannels.

\subsection{Channel Capacity and Majorization}

\begin{figure}[!t]
\centering
\includegraphics[width=3.1in]{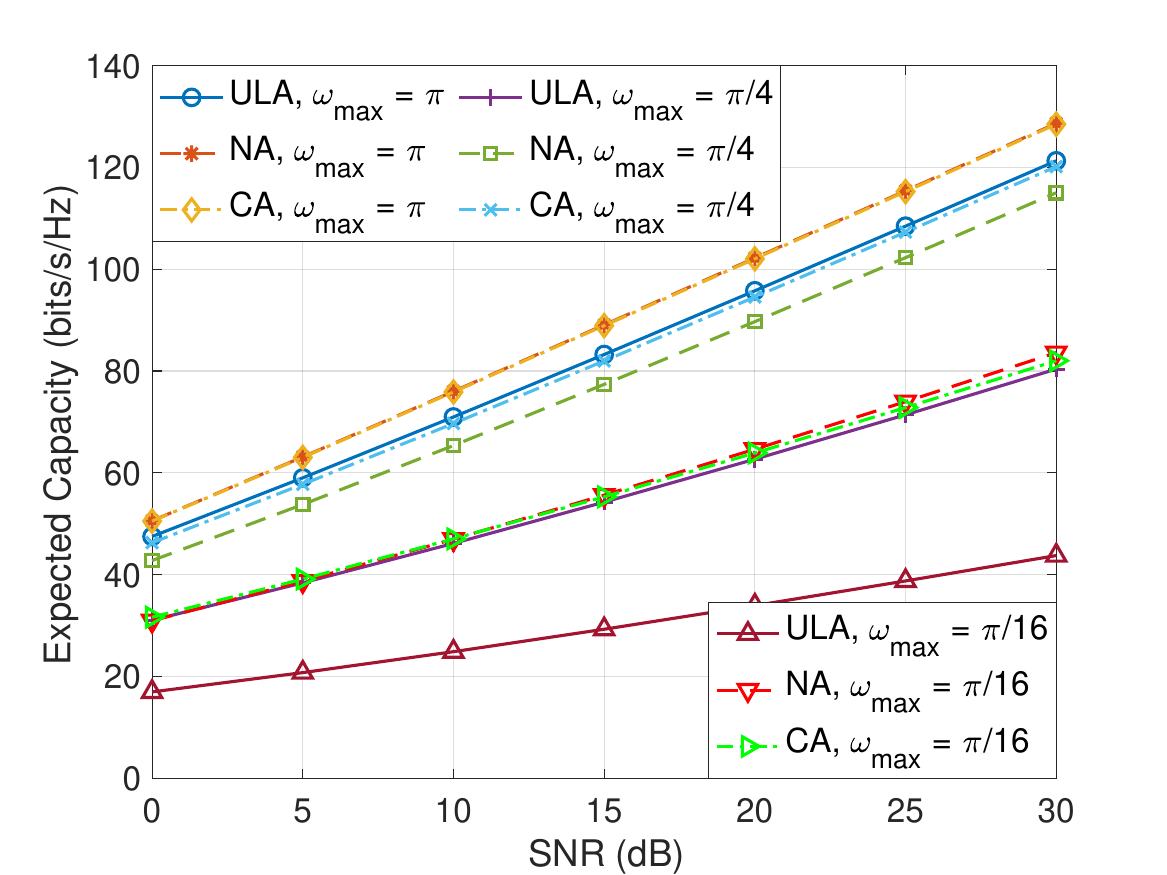}
\vspace{-3mm}
\caption{Expected (ergodic) channel capacity of the MIMO communication channel $\mathbf{H}_{\mathrm{c}}$ in \eqref{eq:H_c_k} in the spatial domain only as the SNR varies.
}
\label{fig:capacity_vs_SNR}
\end{figure}

\begin{figure}[!t]
\centering
\includegraphics[width=3.1in]{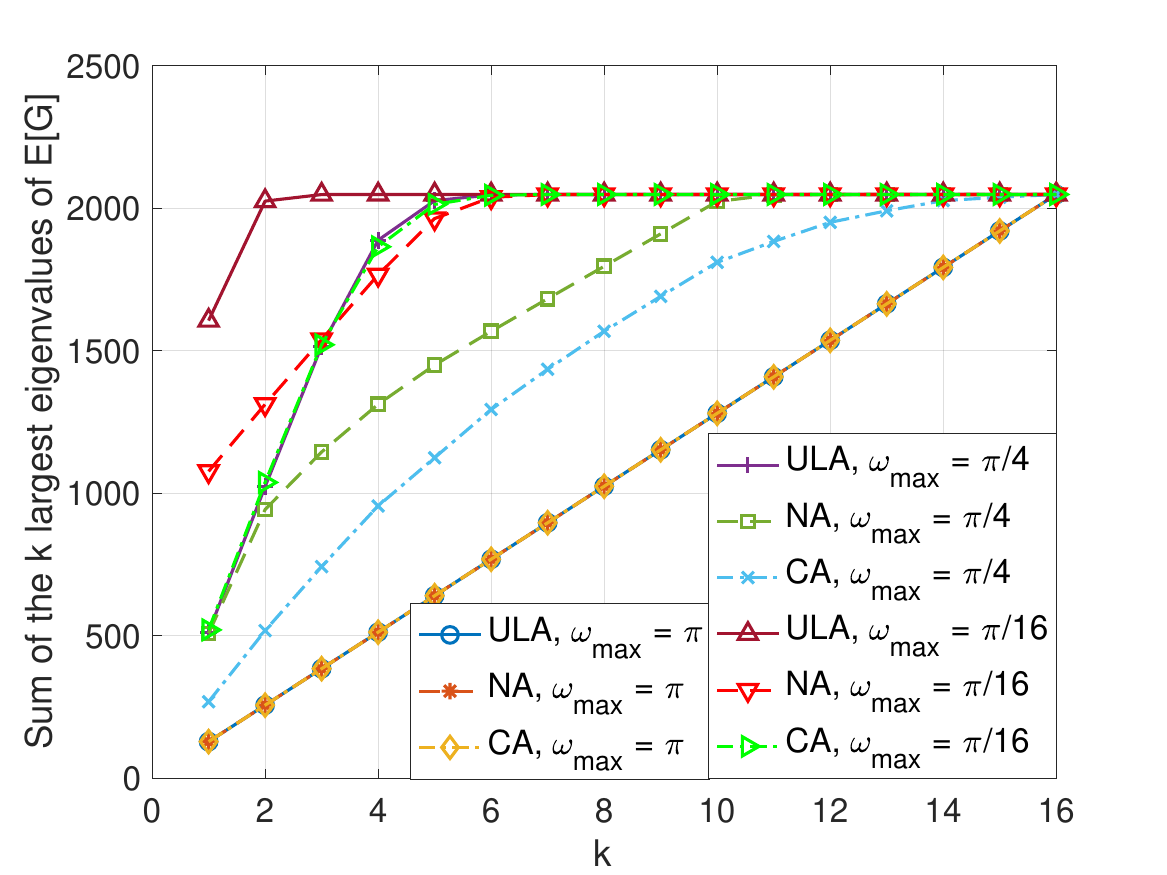}
\vspace{-3mm}
\caption{Sum of the $k$ largest eigenvalues of the expected Gram matrix $\E[\mathbf{G}]$ in the majorization criterion (\ref{eq:mjrz_criterion_G}). The SNR is fixed at 0 dB.
}
\label{fig:fk_vs_k}
\end{figure}

\begin{table}[!t]
% \scriptsize
\renewcommand{\arraystretch}{1.6}
\caption{Testing Results of the Majorization Criterion (Y: Yes)}
\label{tbl:mjrz_criterion}
\centering
\begin{tabular}{c|c|c|c|c|c|c|c}
\hline
$\omega_{\text{max}}$ & $\frac{\pi}{64}$ & $\frac{\pi}{32}$ & $\frac{\pi}{16}$ & $\frac{\pi}{8}$ & $\frac{\pi}{4}$ & $\frac{\pi}{2}$ & $\pi$ \\
\hline
NA majorized by ULA? & Y & Y & Y & Y & Y & Y & Y \\
CA majorized by ULA? & Y & Y & Y & Y & Y & Y & Y \\
\hline
\end{tabular}
\end{table}

In Fig. \ref{fig:capacity_vs_SNR}, we show the expected (ergodic) capacity, averaged over 1000 Monte Carlo trials, of the MIMO communication channel $\mathbf{H}_{\mathrm{c}}$ in \eqref{eq:H_c_k} in the spatial domain only as the SNR varies. A ULA, NA, and CA are compared. We consider $\omega_{\text{max}} = \pi, \pi/4, \pi/16$. For a fixed $\omega_{\text{max}}$, the capacity of the two SAs are consistently larger than that of ULA. Moreover, for a smaller $\omega_{\text{max}}$, corresponding to more tightly clustered paths in the angle domain that are typical in mmWave channels, the capacity gain of each SA over the ULA is larger.
The higher capacity of each SA is perfectly predicted by our majorization method. To see this, in Fig. \ref{fig:fk_vs_k}, we plot the sum of the $k$ largest eigenvalues of the expected Gram matrix $\E[\mathbf{G}]$ for each SA and ULA in the majorization criterion (\ref{eq:mjrz_criterion_G}) when SNR is fixed at 0 dB. For a fixed $\omega_{\text{max}}$, each of the NA and CA is majorized by the ULA, as summarized in Table \ref{tbl:mjrz_criterion} (the results for $\omega_{\text{max}} = \pi/2, \pi/8, \pi/32, \pi/64$ are omitted in Fig. \ref{fig:fk_vs_k} to ensure visual clarity). While all array geometries yield identical expected eigenvalues when $\omega_{\text{max}} = \pi$ (see (\ref{eq:whole_unif})), simulations show that SAs still slightly outperform the ULA in terms of ergodic capacity.

\begin{figure}[!t]
\centering
\includegraphics[width=3.1in]{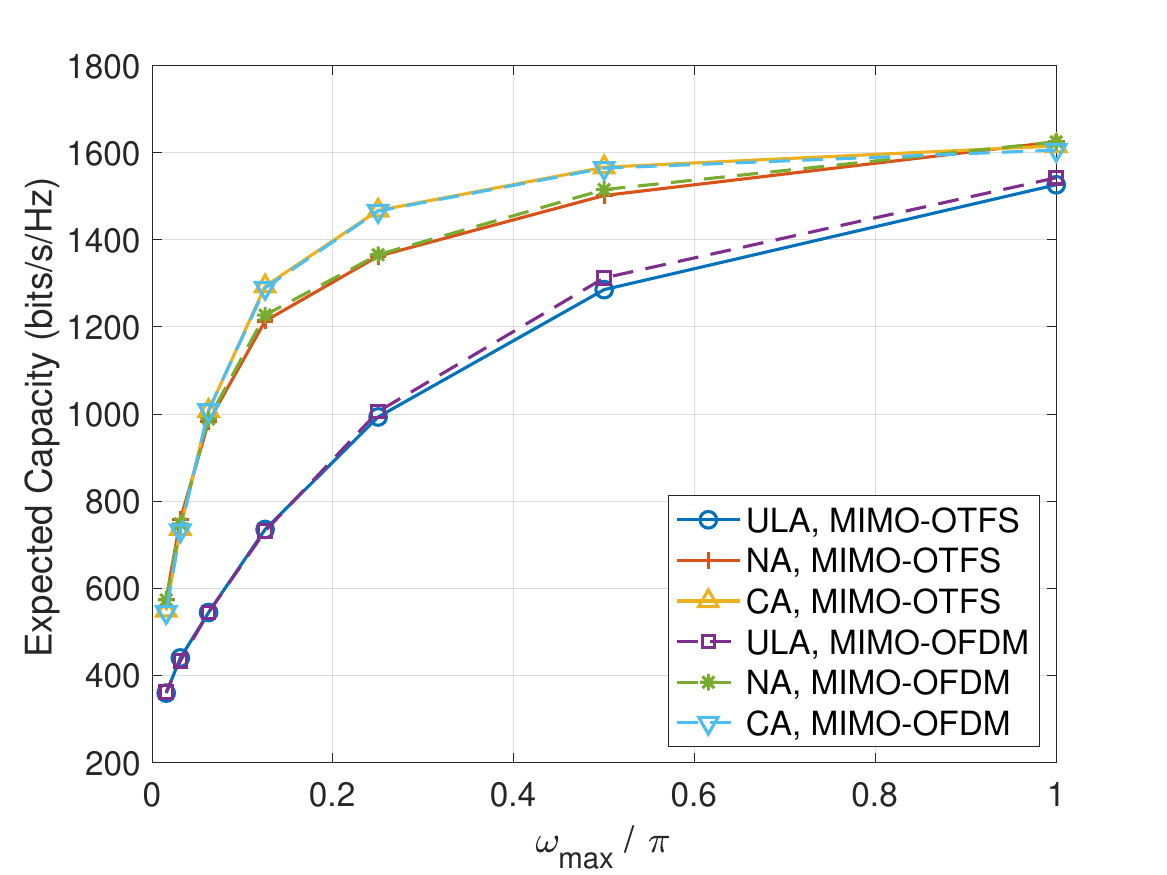}
\vspace{-3mm}
\caption{Expected (ergodic) channel capacity of the MIMO-OTFS channel $\mathbf{H}_{\mathrm{DD,c}}$ and MIMO-OFDM channel $\mathbf{H}_{\mathrm{TF,c}}$ in (\ref{eq:Hzc}) as $\omega_{\text{max}}$ varies. The SNR is fixed at 0 dB.
}
\label{fig:capacity_OTFS_OFDM_vs_wmax}
\end{figure}

\begin{figure}[!t]
\centering
\includegraphics[width=3.1in]{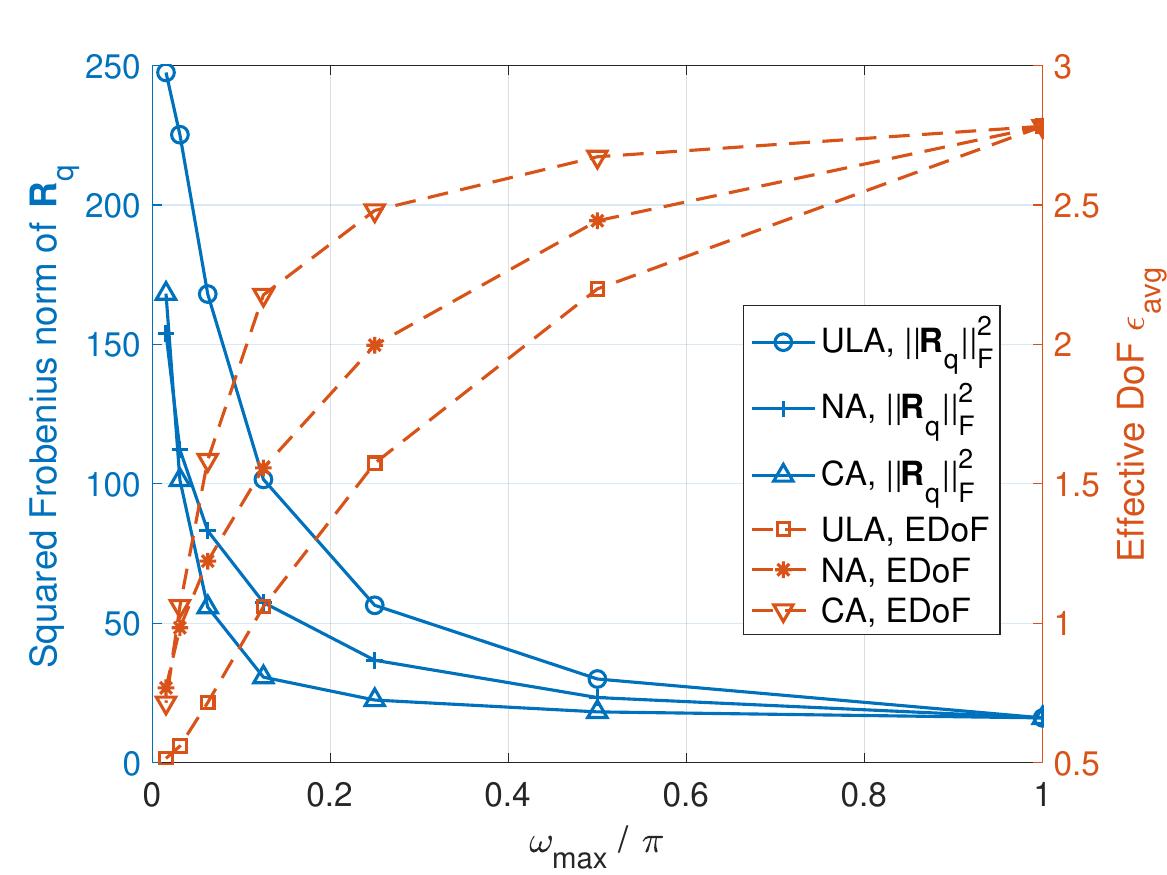}
\vspace{-3mm}
\caption{The squared Frobenius norm of the matrix $\mathbf{R}_q$, $q \in \{\mathrm{t},\mathrm{r}\}$, and average EDoF $\epsilon_{\mathrm{avg}}$ in (\ref{eq:EDoF_avg}) as $\omega_{\text{max}}$ varies. The SNR is fixed at 0 dB.
}
\label{fig:R_F2_EDoF_vs_wmax}
\end{figure}

The previous results for the MIMO channel in the spatial domain only are also valid for MIMO-OTFS and MIMO-OFDM systems. If we plot the ergodic capacity of MIMO-OTFS and MIMO-OFDM as SNR varies, the overall trend will be similar to that in Fig. \ref{fig:capacity_vs_SNR}. To offer more information, we plot, instead, the ergodic capacity of MIMO-OTFS and MIMO-OFDM as $\omega_{\text{max}}$ varies in Fig. \ref{fig:capacity_OTFS_OFDM_vs_wmax}. The SNR is fixed at 0 dB. For a fixed array geometry, MIMO-OTFS and MIMO-OFDM yield almost the same capacity, as also discussed in Remark~\ref{rmk:same_capacity}. Also, both NA and CA perform better than the ULA for all $\omega_{\text{max}}$.
As explained in Sec. \ref{ssc:edof}, while the majorization criterion (\ref{eq:mjrz_criterion}) is the theoretical ``on-off" condition for capacity improvement, $\|\mathbf{R}_q\|_F^2$, $q \in \{\mathrm{t},\mathrm{r}\}$, and average EDoF $\epsilon_{\mathrm{avg}}$ in (\ref{eq:EDoF_avg}) are quantifiable proxies that enable practical comparison between different array geometries. As shown in Fig. \ref{fig:R_F2_EDoF_vs_wmax}, $\|\mathbf{R}_q\|_F^2$ is roughly inversely proportional to the capacity, and $\epsilon_{\mathrm{avg}}$ is roughly proportional to the capacity.

\subsection{Angle CRB}

\begin{figure}[!t]
\centering
\includegraphics[width=3.1in]{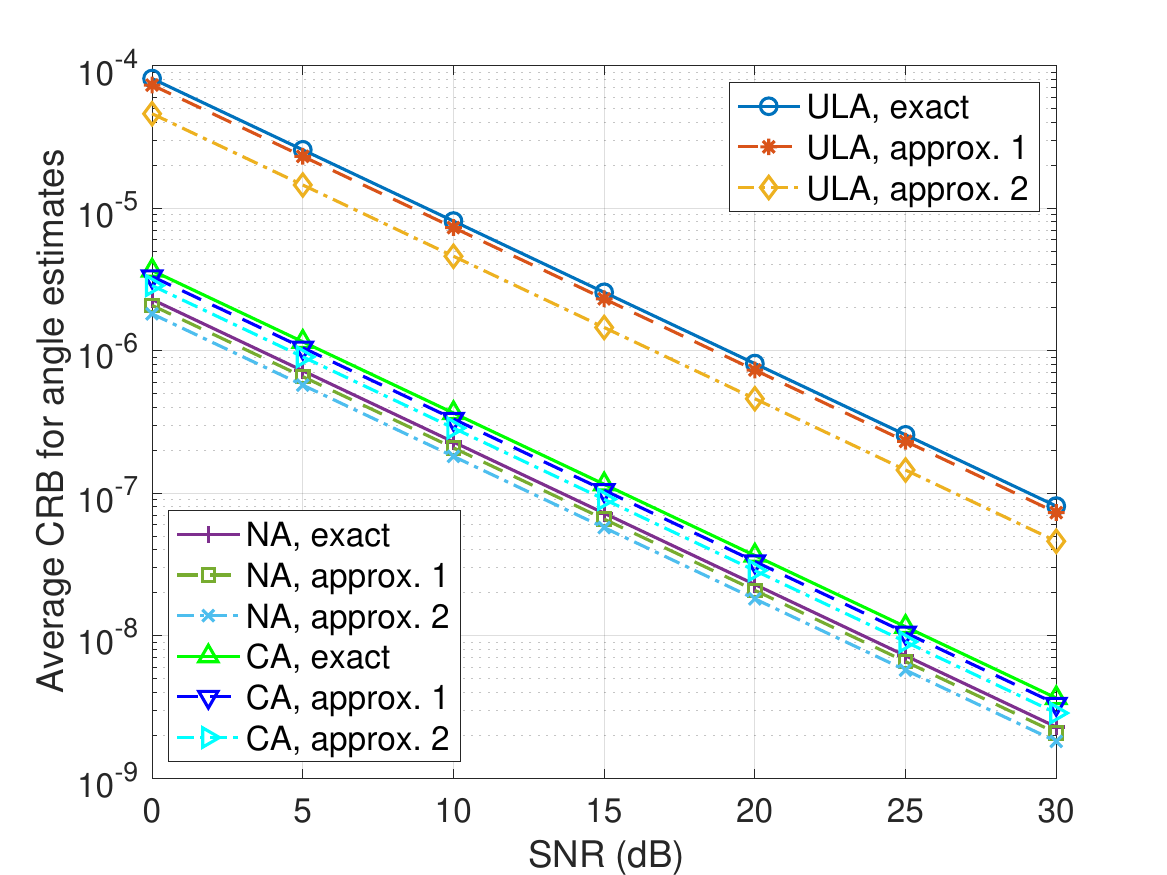}
\vspace{-3mm}
\caption{Angle CRB of the MIMO system \eqref{eq:y_s_k} under monostatic sensing as the SNR varies. The exact CRB is evaluated with (\ref{eq:CRB_exact}), where $\hat{\mathbf{R}}_{\mathrm{s}}$ is computed from (\ref{eq:R_s}). The approximation 1 of CRB is evaluated with (\ref{eq:CRB_Rs_approx}) and (\ref{eq:h_exact}). The approximation 2 of CRB is evaluated with (\ref{eq:CRB_Rs_approx}) and (\ref{eq:h_approx}).
}
\label{fig:AngleCRB_vs_SNR}
\end{figure}

\begin{figure}[!t]
\centering
\includegraphics[width=3.1in]{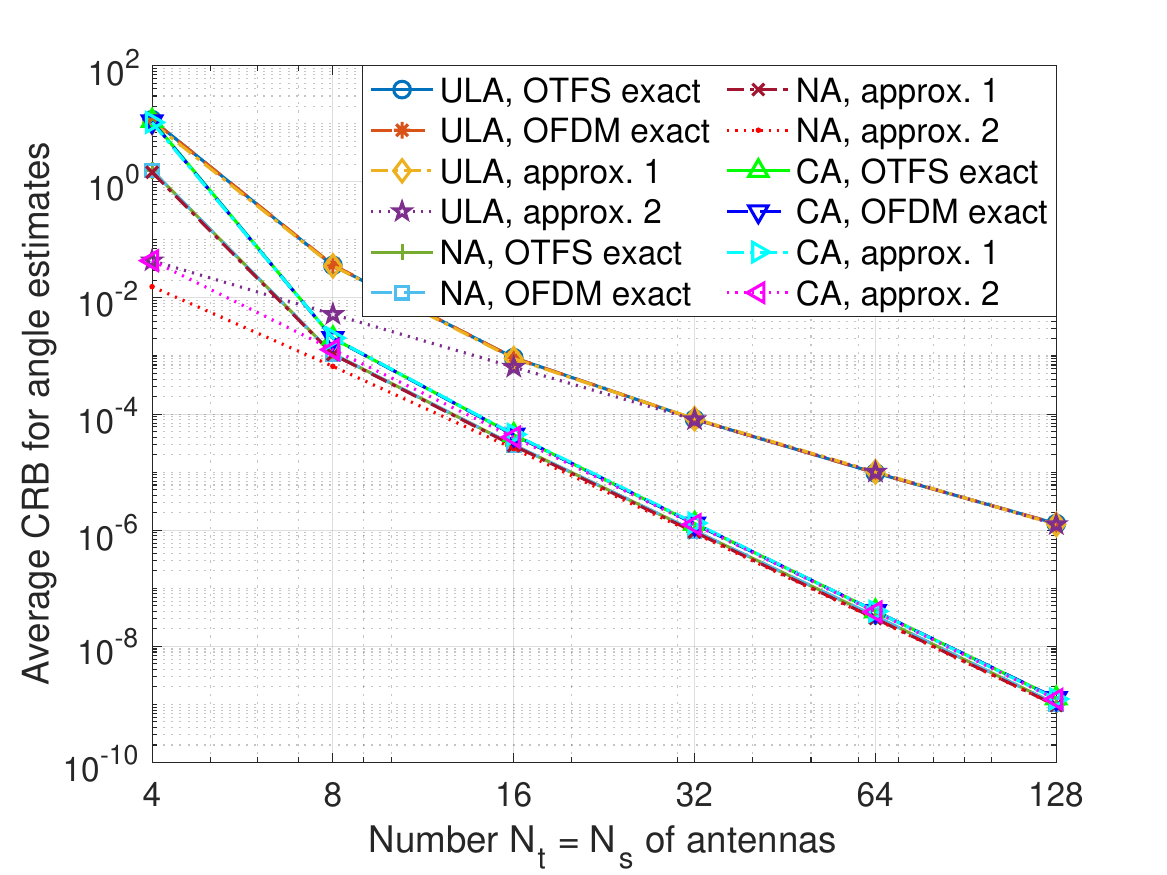}
\vspace{-3mm}
\caption{Angle CRB of the MIMO-OTFS and MIMO-OFDM systems (\ref{eq:MIMO_WF_s}) under monostatic sensing as the number $N_{\mathrm{t}} = N_{\mathrm{s}}$ of antennas varies. The SNR is fixed at 0 dB.
}
\label{fig:AngleCRB_vs_N}
\end{figure}

In Fig. \ref{fig:AngleCRB_vs_SNR}, we show the angle CRB measured in $\omega_{\mathrm{s},i}$ of the MIMO system \eqref{eq:y_s_k} under monostatic sensing as the SNR varies.
Since the conditional CRB model \cite{LiuCRB,StoicaStochasticCRB} is used, meaning the CRB is conditioned on the transmitted symbols $\mathbf{x}[k]$, we run 1000 Monte Carlo trials to get the average angle CRB.
We assume $\omega_{\text{max}} = \pi$.
We again compare the ULA, NA, and CA.
For each array, we plot 3 CRB curves: the exact CRB evaluated with (\ref{eq:CRB_exact}), where $\hat{\mathbf{R}}_{\mathrm{s}}$ is computed from (\ref{eq:R_s}), the approximation 1 of CRB evaluated with (\ref{eq:CRB_Rs_approx}) and (\ref{eq:h_exact}), and the approximation 2 of CRB evaluated with (\ref{eq:CRB_Rs_approx}) and (\ref{eq:h_approx}).
For each array, approximation 1 of CRB coincides with the exact CRB, which implies that (\ref{eq:R_s_approx}) is a very good approximation even when $K=32$ is not very large.
While there is a small gap between approximation 2 of CRB and the exact CRB, it does not alter the main result that the two SAs, NA and CA, always yield much smaller CRBs than the ULA.
In Fig. \ref{fig:AngleCRB_vs_N}, we show the angle CRB of the MIMO-OTFS and MIMO-OFDM systems \eqref{eq:MIMO_WF_s} as the number $N_{\mathrm{t}} = N_{\mathrm{s}}$ of antennas varies while the SNR is fixed at 0 dB.
For each array, we plot 4 CRB curves: the two exact CRBs evaluated with (\ref{eq:CRB__WF_exact}), where $\hat{\mathbf{R}}_{z,\mathrm{s}}$, $z \in \{\mathrm{DD}, \mathrm{TF}\}$, are computed from (\ref{eq:R_zs}), the approximation 1 of CRB evaluated with (\ref{eq:CRB_Rs_WF_approx}) and (\ref{eq:h_exact}), and the approximation 2 of CRB evaluated with (\ref{eq:CRB_Rs_WF_approx}) and (\ref{eq:h_approx}).
Again, the NA and CA always yield smaller CRBs than the ULA.
For each array, the exact CRBs of MIMO-OTFS and MIMO-OFDM are almost the same, which is consistent with our analysis.
Also, approximation 1 of CRB coincides with the exact CRBs, which implies that (\ref{eq:R_DDs_approx}) and (\ref{eq:R_TFs_approx}) are very good approximations even when $L=16$ is not very large.
Besides, the gap between approximation 2 and the other three curves decreases and vanishes in the end as the number of antennas gets larger. This suggests that approximation 2 and our order analysis of CRB in (\ref{eq:CRB_ULA}) and (\ref{eq:CRB_SA}) give good estimates of the CRB asymptotically.
In Fig. \ref{fig:AngleCRB_vs_N}, we can also see that if $N_{\mathrm{s}}$ is not too small (say, larger than 8), the CRB indeed decreases roughly as $O(1/N_{\mathrm{s}}^3)$ and $O(1/N_{\mathrm{s}}^5)$ for ULA and SA, respectively.

\subsection{CRB-Capacity Boundary}

\begin{figure}[!t]
\centering
\includegraphics[width=3.1in]{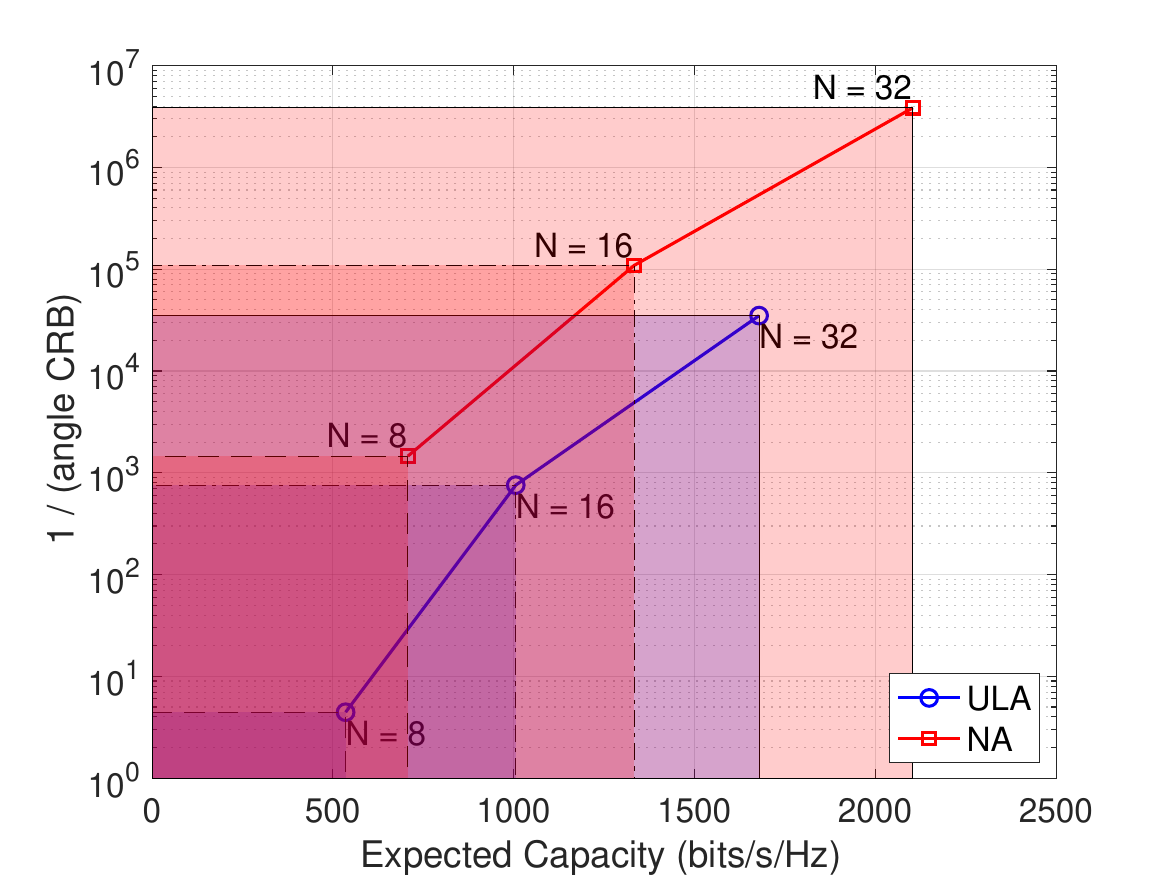}
\vspace{-3mm}
\caption{The achievable CRB-capacity region of the MIMO-OTFS ISAC system under monostatic sensing as the numbers $N_{\mathrm{t}} = N_{\mathrm{r}} = N_{\mathrm{s}} = N$ of antennas vary. The shaded areas represent the achievable regions of ULA and NA, respectively, for each number of antennas.
}
\label{fig:AngleCRB_vs_cpacity}
\end{figure}

In Fig. \ref{fig:AngleCRB_vs_cpacity}, we show the achievable CRB-capacity region of the MIMO-OTFS ISAC system under monostatic sensing as the numbers $N_{\mathrm{t}} = N_{\mathrm{r}} = N_{\mathrm{s}} = N$ of antennas vary.
For clear presentation, the results of MIMO-OFDM, which are close to those of MIMO-OTFS, are omitted.
Similarly, we only show the results of NA but not of CA for clarity.
The SNR is 0 dB, and $\omega_{\text{max}} = \pi/4$.
We see that for the same number of antennas, the achievable CRB-capacity region of NA is strictly larger than that of ULA.
Specifically, both the inverse of the angle CRB and the ergodic channel capacity of NA are higher than those of ULA.

\section{Conclusion}

In this paper, we have studied the fundamental performance limits of MIMO ISAC systems from a SA perspective. Our analytical results, supported by stochastic majorization, demonstrate that SAs significantly outperform conventional ULAs by providing higher ergodic capacity through spatial eigenvalue decorrelation and superior sensing precision due to an expanded physical aperture.
Specifically, the sensing gain over a ULA scales quadratically with the number of antennas.
We show that these conclusions also hold for MIMO-OTFS and MIMO-OFDM ISAC systems.
Future work will investigate optimized SA geometries that maximize communication capacity and minimize the angle-estimation CRB.
Future work will also consider how array configurations affect the estimation accuracy of range and Doppler parameters, providing a more comprehensive framework for future ISAC system design.

% \section*{Acknowledgments}

\appendices
\section{Proof of Theorem~\ref{thm:EDoF}} \label{adx:avgEDoF}
%To characterize the average behavior of the EDoF, we consider $\E[\sum_{i=1}^{D_{\mathrm{c}}} \lambda_i(\mathbf{G})]$ and $\E[\sum_{i=1}^{D_{\mathrm{c}}} \lambda_i^2(\mathbf{G})]$.
Using \eqref{eq:E_G} and the cyclic property of trace, we can obtain
\begin{equation}
\E \left[ \sum_{i=1}^{D_{\mathrm{c}}} \lambda_i(\mathbf{G}) \right] = \text{tr}(\mathbf{\E[G]})
= D_{\mathrm{c}} \sigma_{\alpha,\mathrm{c}}^2 N_{\mathrm{t}} N_{\mathrm{r}} \label{eq:sum_of_eig_G}
\end{equation}
for any array geometry.
Next, note that $\sum_{i=1}^{D_{\mathrm{c}}} \lambda_i^2(\mathbf{G}) = \text{tr}(\mathbf{G}\mathbf{G}^H) = \text{tr}(\mathbf{H}_{\mathrm{c}}\mathbf{H}_{\mathrm{c}}^H \mathbf{H}_{\mathrm{c}}\mathbf{H}_{\mathrm{c}}^H)$.
% \begin{equation}
%     \sum_{i=1}^{D_{\mathrm{c}}} \lambda_i^2(\mathbf{G}) = \text{tr}(\mathbf{G}\mathbf{G}^H) = \text{tr}(\mathbf{H}\mathbf{H}^H \mathbf{H}\mathbf{H}^H).
% \end{equation}
Substituting the summation for $\mathbf{H}_{\mathrm{c}}$ in \eqref{eq:H_c_k}, we have
\begin{IEEEeqnarray}{rCl}
    \sum_{i=1}^{D_{\mathrm{c}}} \lambda_i^2(\mathbf{G}) &=& \sum_{i} \sum_{j} \sum_{k} \sum_{m} \alpha_{\mathrm{c},i} \alpha_{\mathrm{c},j}^* \alpha_{\mathrm{c},k} \alpha_{\mathrm{c},m}^* \IEEEnonumber \\
    && {} \cdot \text{tr} \left( (\mathbf{a}_{\mathrm{r},i} \mathbf{a}_{\mathrm{t},i}^H \mathbf{a}_{\mathrm{t},j} \mathbf{a}_{\mathrm{r},j}^H) (\mathbf{a}_{\mathrm{r},k} \mathbf{a}_{\mathrm{t},k}^H \mathbf{a}_{\mathrm{t},m} \mathbf{a}_{\mathrm{r},m}^H) \right), \IEEEeqnarraynumspace
\end{IEEEeqnarray}
where we use the notation $\mathbf{a}_{\mathrm{r},i}$ to denote $\mathbf{a}_{\mathrm{r}}(\omega_{\mathrm{r},i})$.
In what follows, we use the fourth-moment properties of zero-mean i.i.d. circularly symmetric complex Gaussian (CSCG) path gains $\alpha_{\mathrm{c},i} \sim \mathcal{CN}(0, \sigma_{\alpha,\mathrm{c}}^2)$. The expectation $\E[\alpha_{\mathrm{c},i} \alpha_{\mathrm{c},j}^* \alpha_{\mathrm{c},k} \alpha_{\mathrm{c},m}^*]$ vanishes unless indices are paired such that $i=j, k=m$ or $i=m, j=k$.
When $i=j=k=m$, for each path $i$, the trace term becomes
\begin{align}
    \text{tr}(\mathbf{a}_{\mathrm{r},i} \mathbf{a}_{\mathrm{t},i}^H \mathbf{a}_{\mathrm{t},i} \mathbf{a}_{\mathrm{r},i}^H \mathbf{a}_{\mathrm{r},i} \mathbf{a}_{\mathrm{t},i}^H \mathbf{a}_{\mathrm{t},i} \mathbf{a}_{\mathrm{r},i}^H) &= N_{\mathrm{t}}^2 N_{\mathrm{r}}^2.
\end{align}
When $i=j, k=m, i \neq k$, for pairs of distinct paths $i$ and $k$, the trace term simplifies using the cyclic property of the trace:
\begin{align}
    \text{tr}(\mathbf{a}_{\mathrm{r},i} \mathbf{a}_{\mathrm{t},i}^H \mathbf{a}_{\mathrm{t},i} \mathbf{a}_{\mathrm{r},i}^H \mathbf{a}_{\mathrm{r},k} \mathbf{a}_{\mathrm{t},k}^H \mathbf{a}_{\mathrm{t},k} \mathbf{a}_{\mathrm{r},k}^H) &= N_{\mathrm{t}}^2 | \mathbf{a}_{\mathrm{r},i}^H \mathbf{a}_{\mathrm{r},k} |^2.
\end{align}
Moreover, when $i=m, j=k, i \neq k$, for pairs of distinct paths $i$ and $j$, the trace term can be derived similarly.
Combining the three cases, we obtain
\begin{IEEEeqnarray}{rCl}
    \E \left[ \sum_{i=1}^{D_{\mathrm{c}}} \lambda_i^2(\mathbf{G}) \right] &=& 2 D_{\mathrm{c}} \sigma_{\alpha,\mathrm{c}}^4 N_{\mathrm{r}}^2 N_{\mathrm{t}}^2 \IEEEnonumber \\
    && {} + \sigma_{\alpha,\mathrm{c}}^4 \sum_{i \neq j} N_{\mathrm{t}}^2 \E [ | \mathbf{a}_{\mathrm{r}}^H(\omega_{\mathrm{r},i}) \mathbf{a}_{\mathrm{r}}(\omega_{\mathrm{r},j}) |^2 ] \IEEEnonumber \\
    && {} + \sigma_{\alpha,\mathrm{c}}^4 \sum_{i \neq j} N_{\mathrm{r}}^2 \E [ | \mathbf{a}_{\mathrm{t}}^H(\omega_{\mathrm{t},i}) \mathbf{a}_{\mathrm{t}}(\omega_{\mathrm{t},j}) |^2 ], \IEEEeqnarraynumspace \label{eq:sum_of_sqr_eig_G}
\end{IEEEeqnarray}
where we use the fact that $\E[|\alpha_{\mathrm{c},i}|^4] = 2 \sigma_{\alpha,\mathrm{c}}^4$ for CSCG path gains.
Thus, the expected sum of squared eigenvalues depends on the average spatial correlation $\E[| \mathbf{a}_q^H(\omega_{q,i}) \mathbf{a}_q(\omega_{q,j}) |^2]$, $q\in\{\mathrm{t},\mathrm{r}\}$, over the angular domain.
Note that if we vary $\omega_{q,i}$ while fixing $\omega_{q,j}$, $| \mathbf{a}_q^H(\omega_{q,i}) \mathbf{a}_q(\omega_{q,j}) |^2$ is actually the \textit{beam pattern} of the array.
We can derive that
\begin{IEEEeqnarray}{rCl}
    | \mathbf{a}_q^H(\omega_{q,i}) \mathbf{a}_q(\omega_{q,j}) |^2 &= \sum_{k=0}^{N_q-1} \sum_{m=0}^{N_q-1} e^{j(\omega_{q,j} - \omega_{q,i})(n_{q,k} - n_{q,m})}. \IEEEnonumber
\end{IEEEeqnarray}
Since all $\omega_{q,i}$ are i.i.d, we obtain
\begin{IEEEeqnarray}{rCl}
    \E [ | \mathbf{a}_q^H(\omega_{q,i}) \mathbf{a}_q(\omega_{q,j}) |^2 ] &=& \sum_{k=0}^{N_q-1} \sum_{m=0}^{N_q-1} \E[e^{j\omega_{q,j}d_{q,k,m}}] \IEEEnonumber \\
    && {} \cdot \E[e^{-j\omega_{q,i}d_{q,k,m}}] \IEEEnonumber \IEEEeqnarraynumspace \\
    &=& \sum_{k=0}^{N_q-1} \sum_{m=0}^{N_q-1} |R_{q,k,m}|^2
    = \|\mathbf{R}_q\|_F^2 \IEEEeqnarraynumspace \label{eq:Rq_F2}
\end{IEEEeqnarray}
where $d_{q,k,m} = n_{q,k} - n_{q,m}$, and the entries of $\mathbf{R}_q$ are
% \begin{IEEEeqnarray}{rCl}
% R_{q,k,m} &=& \E[e^{j\omega (n_{q,k} - n_{q,m})}],
% \end{IEEEeqnarray}
$R_{q,k,m} = \E[e^{j\omega (n_{q,k} - n_{q,m})}],$
as established in \eqref{eq:R_rkm}.
Plugging \eqref{eq:sum_of_eig_G}, \eqref{eq:sum_of_sqr_eig_G}, and \eqref{eq:Rq_F2} into the definition (\ref{eq:EDoFavg}) of $\epsilon_{\mathrm{avg}}$ leads to (\ref{eq:EDoF_avg}),
% \begin{IEEEeqnarray}{rCl}
% \epsilon_{\mathrm{avg}} &=& \frac{D_{\mathrm{c}}}{2 + (D_{\mathrm{c}} - 1)[ \|\mathbf{R}_{\mathrm{r}}\|_F^2 / N_{\mathrm{r}}^2 + \|\mathbf{R}_{\mathrm{t}}\|_F^2 / N_{\mathrm{t}}^2]}, \IEEEeqnarraynumspace
% \end{IEEEeqnarray}
which completes the proof.

\bibliographystyle{IEEEtran}
\bibliography{IEEEabrv,bibliography.bib}

% \newpage
 
% \vspace{11pt}

% \bf{If you include a photo:}\vspace{-33pt}
% \begin{IEEEbiography}[{\includegraphics[width=1in,height=1.25in,clip,keepaspectratio]{fig1}}]{Michael Shell}
% Use $\backslash${\tt{begin\{IEEEbiography\}}} and then for the 1st argument use $\backslash${\tt{includegraphics}} to declare and link the author photo.
% Use the author name as the 3rd argument followed by the biography text.
% \end{IEEEbiography}

% \vspace{11pt}

% \bf{If you will not include a photo:}\vspace{-33pt}
% \begin{IEEEbiographynophoto}{John Doe}
% Use $\backslash${\tt{begin\{IEEEbiographynophoto\}}} and the author name as the argument followed by the biography text.
% \end{IEEEbiographynophoto}

\vfill

\end{document}